\documentclass[10pt,conference]{IEEEtran}
\usepackage{amsmath,amsthm,cite}
\usepackage{amssymb}
\usepackage{fancyvrb}
\usepackage{graphicx}
\usepackage{algorithm}
\usepackage{algorithmic}
\usepackage{url}
\usepackage{balance}
\usepackage{longtable}
\usepackage{tabularx}
\usepackage[english]{babel}
\usepackage{blindtext}
\usepackage{balance}
\usepackage{etoolbox}
\newtheorem{thm}{Theorem}
\newtheorem{lemma}[thm]{Lemma}

\begin{document}
\title{Information Source Detection in the SIR Model: A Sample Path Based Approach}

\author{\IEEEauthorblockN{Kai Zhu and Lei Ying}
\IEEEauthorblockA{School of Electrical, Computer and Energy Engineering\\ Arizona State University\\
Tempe, AZ, United States, 85287\\
Email: kzhu17@asu.edu, lei.ying.2@asu.edu}}

\maketitle
\begin{abstract}
This paper studies the problem of detecting the information source
in a network in which the spread of information follows the popular
Susceptible-Infected-Recovered (SIR) model. We assume all nodes in
the network are in the susceptible state initially except the
information source which is in the infected state. Susceptible nodes
may then be infected by infected nodes, and infected nodes may
recover and will not be infected again after recovery. Given a
snapshot of the network, from which we know all infected nodes but
cannot distinguish susceptible nodes and recovered nodes, the
problem is to find the information source based on the snapshot and
the network topology. We develop a sample path based approach where the estimator of the information
source is chosen to be the root node associated with the sample path
that most likely leads to the observed snapshot. We prove for
infinite-trees, the estimator is a node that minimizes the
maximum distance to the infected nodes. A reverse-infection algorithm is proposed to find such an estimator in general
graphs. We prove that for $g$-regular trees such that $gq>1,$ where
$g$ is the node degree and $q$ is the infection probability, the
estimator is within a constant distance from the actual source with a high probability, independent of the
number of infected nodes and the time the snapshot is taken. Our simulation results show that for tree networks, the estimator produced by the
reverse-infection algorithm is closer to the actual source than the
one identified by the closeness centrality heuristic. We then further evaluate the performance of the reverse infection algorithm on several real world networks.
\end{abstract}

\section{Introduction}
Diffusion processes in networks refer to the spread of information
throughout the networks, and have been widely used to model many
real-world phenomena such as the outbreak of epidemics, the
spreading of gossips over online social networks, the spreading of
computer virus over the Internet, and the adoption of innovations.
Important properties of diffusion processes such as the outbreak
thresholds \cite{MooNew_00} and the impact of network topologies
\cite{GanMasTow_05} have been intensively studied.

In this paper, we are interested in the reverse of the diffusion
problem: given a snapshot of the diffusion process at time $t,$ can
we tell which node is the source of the diffusion? The answer to
this problem has many important applications, and can help us answer
the following questions: who is the rumor source in online social
networks? which computer is the first one infected by a computer
virus? who is the one who uploaded contraband materials to the
Internet? and where is the source of an epidemic?

We call this problem information source detection problem. This
information source detection problem has been studied in
\cite{ShaZam_10,ShaZam_11,ShaZam_12} under the
Susceptible-Infected (SI) model, in which susceptible nodes may be
infected but infected nodes cannot recover. The authors formulated
the problem as a maximum likelihood estimation (MLE) problem, and
developed novel algorithms to detect the source.

In this paper, we adopt the Susceptible-Infected-Recovered (SIR)
model, a standard model of epidemics \cite{Bai_75,EasKle_10}.
The network is assumed to be an undirected graph and each node in
the network has three possible states: susceptible ($S$), infected
($I$), and recovered ($R$). Nodes in state $S$ can be infected and
change to state $I$, and nodes in state $I$ can recover and change
to state $R$. Recovered nodes cannot be infected again. We assume
that initially all nodes are in the susceptible state except one
infected node (called the information source). The information
source then infects its neighbors, and the information starts to
spread in the network.  Now given a snapshot of the network, in
which we can identify infected nodes and healthy (susceptible and
recovered) nodes (we assume susceptible nodes and recovered nodes
are indistinguishable), the question is which node is the
information source.

We remark that it is very important to take recovery into consideration since recovery can happen due to various
reasons in practice. For example, a contraband material uploader may delete the file, a
computer may recover from a virus attack after anti-virus software
removes the virus, and a user may delete the rumor from her/his
blog. In order to solve the information source detection problem in
these scenarios, we study the SIR model in this paper, which makes
the problem significantly more challenging than that in the SI model
as we will explain in the related work section.

\subsection{Main Results}
The main results of this paper are summarized below.

\begin{itemize}
\item Similar to the SI model, the information source detection problem can be
  formalized as an MLE problem. Unfortunately, to solve the MLE
 problem, we need to consider all possible infection sample paths, and for each sample path,
  we need to specify the infection time and recovery time for each healthy node
  and the infection time
for each infected node, so the number of possible sample paths is at
the order of $\Omega(t^N),$ where $N$ is the network size and $t$ is
the time the snapshot is obtained. Therefore, the MLE problem is
difficult to solve even when $t$ is known. The problem becomes
much harder when $t$ is unknown, which is the assumption of this
paper. To overcome this difficulty,
  we propose a sample path based approach. We propose to find the sample
  path which most likely leads to the observed snapshot and view the source associated with
 that sample path as the information source.
 We call this problem optimal sample path detection problem.
We investigate the structure properties of the optimal sample path
in trees. Defining the infection eccentricity of a node to be the
maximum distance from the node to infected nodes, we prove that the
source node of the optimal sample path is the node with the minimum
infection eccentricity. Since a node with the minimum eccentricity
in a graph is called the Jordan center, we call the nodes with the
minimum infection eccentricity the Jordan infection centers. Therefore, the sample path based estimator is one of the Jordan infection centers.

\item We propose a low complexity
  algorithm, called reverse infection algorithm, to find the sample path based estimator in general graphs.
 In the algorithm, each infected node broadcasts its identity in the network, the node
who first collect all identities of infected nodes declares itself
as the information source, breaking ties based on the sum of distances to infected nodes.  The running time
of this algorithm is equal to the minimum infection eccentricity,
and the number of messages each node receives/sends at each iteration is bounded by
 the degree of the node.

\item We analyze the performance of the reverse infection algorithm on $g$-regular trees,
and show that the algorithm can output
  a node within a constant distance from the actual source with a high probability, independent of the
  number of infected nodes and the time the snapshot is taken.

\item We conduct extensive simulations over various
  networks to verify the performance of the reverse infection algorithm. The
  detection rate over regular trees is found to be around $60\%$,
  and is higher than that of the infection closeness centrality (or called distance centrality) heuristic. The infection closeness of a node is defined to be the
inverse of the sum of distances to infected nodes and the infection closeness centrality
heuristic is to claim the node with the maximum infection closeness
as the source. Note that in \cite{ShaZam_10,ShaZam_11,ShaZam_12}, the
authors proved the node with the maximum infection closeness is the MLE on regular trees. For real world
networks, our experiments also show that the reverse infection
algorithm outperforms random guesses significantly. We then further evaluate the performance of the reverse infection algorithm on several real world networks.
\end{itemize}

\subsection{Related Work}
There have been extensive studies on the spread
of epidemics in networks based on the SIR model (see
\cite{New_02,PasVes_01,MooNew_00,GanMasTow_05}
and references within). The work most related to this paper is
\cite{ShaZam_10,ShaZam_11,ShaZam_12}, in which the information source
detection problem was studied under the SI model. \cite{LuoTay_12,LuoTayLen_12} considers the problem of detecting multiple information sources under the SI model. This paper considers the SIR model, where
infection nodes may recover, which can occur in many practical
scenarios as we have explained. Because of node recovery, the information source
detection problem under the SIR model differs significantly from that under the SI
model. The differences are summarized below.
\begin{itemize}
\item The set of possible sources in the SI model \cite{ShaZam_10,ShaZam_11,ShaZam_12}
 is restricted to the set of infected nodes. In the SIR model, all
nodes are possible information sources because we assume susceptible
nodes and recovered nodes are indistinguishable and a healthy node
may be a recovered node so can be the information source. Therefore,
the number of candidate sources is much larger in the SIR model than
that in the SI model.

\item A key observation in \cite{ShaZam_10,ShaZam_11,ShaZam_12} is that on regular trees, all permitted permutations of infection sequences (a infection
sequence specifies the order at which nodes are infected) are
equally likely under the SI model. The number of possible permutations from a fixed
root node, therefore, decides the likelihood of the root node being
the source. However, under the SIR model, different infection sequences are associated with different probabilities, so counting the number of permutations are not sufficient.

\item \cite{ShaZam_10,ShaZam_11,ShaZam_12} proved that the node
 with the maximum closeness centrality is the
an MLE on regular-trees. We define the infection closeness
centrality to be the inverse of the sum of distances to infected nodes. Our simulations show
that the sample path based estimator is closer to the actual source
than the nodes with the maximum infection closeness.
\end{itemize}

Other related works include: (1) detecting the first adopter of innovations based on a game theoretical model \cite{SubBer_12} in which the authors derived the MLE but the computational complexity is exponential in the number of nodes, (2) network forensics under the SI model \cite{MilCarSha_12}, where the goal is to distinguish an epidemic infection from a random infection, and (3) geospatial abduction problems (see \cite{ShaSubSap_11,ShaSub_11} and references within).

\section{Problem Formulation}

\subsection{The SIR Model for Information Propagation}
Consider an undirected graph $G=\{\cal{V},\cal{E}\}$, where $\cal V$
is the set of nodes and $\cal E$ is the set of (undirected) edges.
Each node $v\in \cal{V}$ has three possible states: susceptible ($S$), infected ($I$), and recovered ($R$). We assume a time slotted
system. Nodes change their states at the beginning of each time
slot, and the state of node $v$ in time slot $t$ is denoted by
$X_v(t).$ Initially, all nodes are in state $S$ except node $v^*$
which is in state $I$ and is the information source. At the
beginning of each time slot, each infected node infects each of its
susceptible neighbors with probability $q,$ independent of other
nodes, i.e., a susceptible node is infected with probability
$1-(1-q)^{n}$ if it has $n$ infected neighbors. Each infected node
recovers with probability $p$, i.e., its state changes from $I$ to
$R$ with probability $p.$ In addition, we assume a recovered node
cannot be infected again. Since whether a node gets infected only
depends on the states of its neighbors and whether a node becomes a
recovered node only depends on its own state in the previous time
slot, the infection process can be modeled as a discrete time Markov
chain ${\bf X}(t)$ where ${\bf X}(t)=\{X_v(t),v\in{\cal V}\}$ is the states of all the nodes
at time slot $t.$ The initial state of this Markov chain is $X_v(0)=S$
for $v\neq v^*$ and $X_{v^*}(0)=I.$

\subsection{Information Source Detection}
We assume ${\bf X}(t)$ is not fully observable since we cannot
distinguish susceptible nodes and recovered ones. So at time $t$, we
observe ${\bf Y}=\{Y_v, v\in{\cal V}\}$ such that
\[
Y_v=\left\{
               \begin{array}{ll}
                 1, & \hbox{if $v$ is in state $I$;} \\
                 0, & \hbox{if $v$ is in state $S$ or $R$.}
               \end{array}
             \right.
\]
The information source detection problem is to identify $v^*$ given
the graph $G$ and ${\bf Y},$ where $t$ is an
unknown parameter.

\begin{figure}
\begin{centering}
  \includegraphics[width=1\columnwidth]{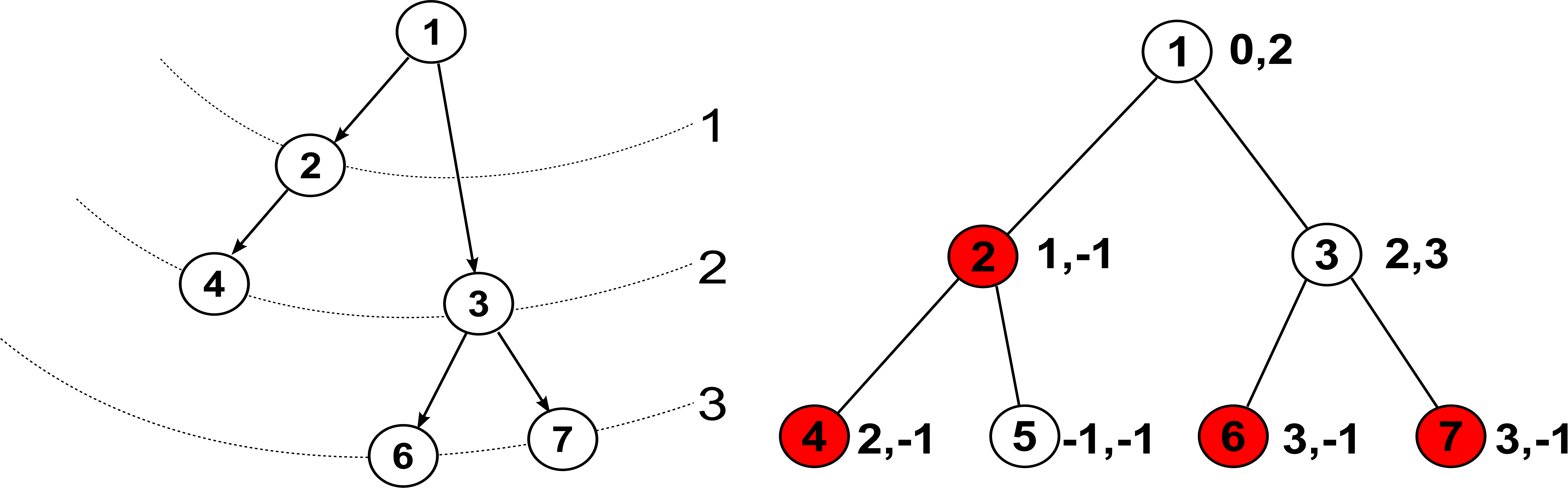}\\
  \caption{An Example of Information Propagation}\label{figure:ip_eg}
  \end{centering}
\end{figure}

Figure \ref{figure:ip_eg} is an example of the infection process.
The left figure shows the information propagation over time. The
nodes on each dotted line are the nodes which are infected at that
time slot, and the arrows indicate where the infection comes from
(e.g., node $4$ is infected by node $2$).

The figure on the right is the network we observe, where the shaded
nodes are infected nodes and others are susceptible or recovered
nodes. The pair of numbers next to each node are the corresponding
infection time and recovery time. For example, node $3$ was infected
at time slot $2$ and recovered at time slot $3$. $-1$ indicates that
the infection or recovery has yet occurred. Note that these two
pieces of information are not available to us, and we include them
in the figure to illustrate the infection and recovery processes. If
we observe the network at the end of time slot $3,$ then the
snapshot of the network is ${\bf Y}=\{0,1,0,1,0,1,1\},$ where the
states are ordered according to the indices of the nodes.

\subsection{Maximum Likelihood Detection}
We define ${\bf X}[0,t]=\{{\bf X}(\tau):0<\tau\leq t\}$ to be a
sample path of the infection process from $0$ to $t.$ In addition,
we define function $F(\cdot)$ such that
\[
F(X_v(t))=\left\{
  \begin{array}{ll}
    1, & \hbox{if $X_v(t)=I$;} \\
    0, & \hbox{otherwise.}
  \end{array}
\right.
\]
We say ${\bf F}(\mathbf{X}[t])={\bf Y}$ if
$F(X_v(t))=Y_v$ for all $v.$ Identifying the information source can
be formulated as a maximum likelihood detection problem as follows:
\[
v^{\dag}\in\hbox{arg} \max_{v\in {\cal V}} \sum_{{\bf X}[0,t]: {\bf
F}({\bf X}(t))={\bf Y}} \Pr({\bf X}[0,t]|v^*=v),
\]
where $\Pr({\bf X}[0,t]|v^*=v)$ is the probability to obtain sample
path ${\bf X}[0,t]$ given the information source is node $v.$

We note the difficulty of solving this maximum likelihood problem is
the curse of dimensionality. For each $v$ such that $Y_v=0,$ we need
to decide its infection time and recovery time (the node is in
susceptible state if the infection time is $>t$), i.e., $O(t^2)$
possible choices; for each $v$ such that $Y_v=1,$ we need to decide
the infection time, i.e., $O(t)$ possible choices. Therefore, even
for a fixed $t,$ the number of possible sample paths is at least at the
order of $t^N,$ where $N$ is the number of nodes in the network. This curse of
dimensionality makes it computationally expensive, if not impossible,
to solve the maximum likelihood problem. To overcome this
difficulty, we propose a sample path based approach which is
discussed below.

\subsection{Sample Path Based Detection}
Instead of using the MLE, we propose to
identify the sample path ${\bf X}^*[0, t^*]$ that most likely leads
to ${\bf Y},$ i.e.,
\begin{align}
{\bf X}^*[0,t^*]=\hbox{arg}\max_{t,{\bf X}[0,t]\in {\cal
X}(t)}\Pr\left({\bf X}[0,t]\right),\label{eqn:optsamplepath}
\end{align}
where ${\cal X}(t)=\{{\bf X}[0,t]|{\bf F}({\bf X}(t))={\bf Y}\}.$
The source node associated with ${\bf X}^*[0,t^*]$ is then viewed as the
information source.

\section{Sample Path Based Detection On Tree Networks}
The optimal sample paths for general graphs are still difficult to
obtain. In this section, we focus on tree networks and derive structure
properties of the optimal sample paths.

First, we introduce the definition of eccentricity in graph theory
\cite{Har_91}. The eccentricity $e(v)$ of a vertex $v$ is
the maximum distance between $v$ and any other vertex in the graph.
The Jordan centers of a graph are the nodes which have the minimum
eccentricity. For example, in Figure
\ref{figure:infectioneccentricity}, the eccentricity of node $v_1$
is $4$ and the Jordan center is $v_2,$ whose eccentricity is $3.$

\begin{figure}
\begin{centering}
  \includegraphics[width=3in]{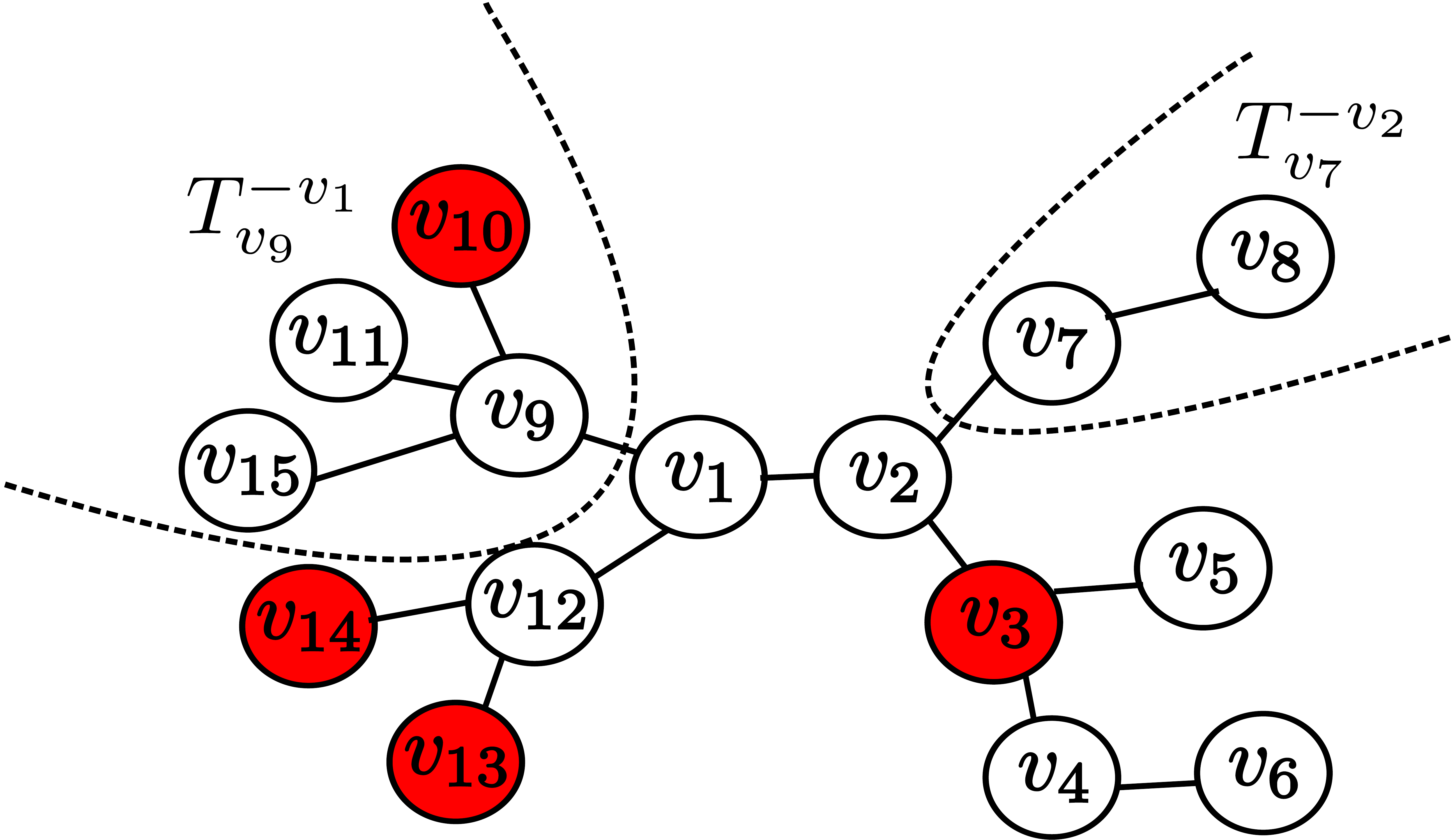}
  \caption{An Example Illustrating the Infection Eccentricity}\label{figure:infectioneccentricity}
  \end{centering}
\end{figure}
Following a similar terminology, we define the infection
eccentricity $\tilde{e}(v)$ given ${\bf Y}$ as the maximum distance
between $v$ and any infected nodes in the graph. Define the Jordan
infection centers of a graph to be the nodes with the minimum
infection eccentricity given ${\bf Y}.$ In Figure
\ref{figure:infectioneccentricity}, nodes $v_3,$ $v_{10},$ $v_{13}$
and $v_{14}$ are observed to be infected. The infection
eccentricities of $v_1,v_2,v_3,v_4$ are $2,3,4,5,$ respectively, and
the Jordan infection center is $v_1.$

We will show that the source associated with the optimal sample path
is a node with the minimum infection eccentricity. We derive this
result using three steps: first, assuming the information source is
$v_r,$ we analyze $t^*_{v_r}$ such that $$t^*_{v_r}=\displaystyle\arg_{t}\max_{t,
{\bf X}[0,t]}\Pr({\bf X}[0,t]|v^*=v_r),$$ i.e., $t^*_{v_r}$ is the
time duration of the optimal sample path in which $v_r$
is the information source. It turns out that $t^*_{v_r}$ equals to
the infection eccentricity of node $v_r.$ Considering Figure \ref{figure:infectioneccentricity} if the source
is $v_1$, then the time duration of the optimal sample path starting from $v_1$ is $2.$

In the second step, we consider two neighboring nodes, say nodes $v_1$ and $v_2.$ We will prove that if $\tilde{e}(v_1)<\tilde{e}(v_2),$ then the optimal sample path rooted at $v_1$ occurs with a higher probability than the optimal sample path rooted at $v_2.$

Finally, at the third step, we will show that given any two nodes $u$
and $v,$ if $v$ has the minimum infection eccentricity and $u$ has a
larger infection eccentricity, then there exists a path from $u$ to
$v$ along which the infection eccentricity monotonically decreases,
which implies that the source of the optimal sample path must be a
Jordan infection center. For example, in Figure
\ref{figure:infectioneccentricity}, node $v_4$ has a larger
infection eccentricity than $v_1$ and $v_4\rightarrow v_3
\rightarrow v_2 \rightarrow v_1$ is the path along which the
infection eccentricity monotonically decreases from $5$ to $2.$

\subsection{The Optimal Time}
\begin{lemma}\label{lem:inequalitytime}
Consider a tree network rooted at $v_r$ and with infinitely many
levels. Assume the information source is the root, and the observed
infection topology is ${\bf Y}$ which contains at least one infected node. If $\tilde{e}(v_r)\leq t_1<t_2,$ then the following inequality
holds
$$\max_{{\bf X}[0,t_1]\in {\cal X}(t_1)} \Pr({\bf X}[0,t_1])> \max_{{\bf X}[0,t_2] \in {\cal X}(t_2)} \Pr({\bf X}[0,t_2]),$$
where ${\cal X}(t)=\{{\bf X}[0,t]|{\bf F}({\bf
X}(t))={\bf Y}\}.$ In addition,
$$t_{v_r}^*=\tilde{e}(v_r)=\max_{u\in{\cal I}}d(v_r,u),$$
where $d(v_r,u)$ is the length of the shortest path between $v_r$
and $u$ and also called the distance between $v_r$ and $u$, and
${\cal I}$ is the set of infected nodes.  \hfill{$\square$}
\end{lemma}

\begin{proof}
 We start from the case where the time difference of two sample paths is
one, i.e., we will show that
\begin{align}
\max_{{\bf X}[0,t]\in {\cal X}(t)} \Pr({\bf X}[0,t])> \max_{{\bf
X}[0,t+1] \in {\cal X}(t+1)} \Pr({\bf X}[0, t+1]).\label{eqn:3}
\end{align}

We divide all possible infection topologies ${\bf Y}$ into countable
subsets $\{{\cal Y}^k\}$ where ${\cal Y}^k$ is the set of infection
topologies where the largest distance from $v_r$ to an infected node
is $k$. ${\cal Y}^0$ is the topology where there is only one
infected node---the root node $v_r$. Note that if no infected node
is observed, no algorithm performs better than a random guess.  To
prove (\ref{eqn:3}), we use induction over $k$.

{\bf Step 1:} First, we consider the case $k=0$. All the sample paths considered in step 1 lead to observation $ {\bf Y}\in{\cal Y}^0.$ We denote by $T_{v_r}$ the tree rooted in
$v_r$ and $T^{-v_r}_u$ the tree rooted at $u$ but without the branch
from $v_r.$ For example, Figure \ref{figure:timeinequality} shows
$T_{v_1}^{-v_r},$ $T_{v_2}^{-v_r},$ $T_{v_3}^{-v_r}$ and
$T_{v_4}^{-v_r}$. The sample path from time slot $0$ to $t$
restricted to $T_u^{-v_r}$ is denoted by ${\bf
X}([0,t],T_u^{-v_r}).$ Furthermore, denote by ${\cal C}(v)$ the set of children of $v.$ We have
\begin{align*}
&\Pr({\bf X}[0,t])\\
&=\Pr(X_{v_r}(s)=I,0\leq s\leq t)\\
&\times\prod_{u\in{\cal C}(v_r)}\Pr({\bf
X}([0,t],T_u^{-v_r})| X_{v_r}(t)=I)\\
&=(1-p)^t\prod_{u\in{\cal C}(v_r)}\Pr({\bf
X}([0,t],T_u^{-v_r})| X_{v_r}(t)=I),
\end{align*}
where the last equality holds since $v_r$ is the only infected node in the network at time $t,$ which requires $X_{v_r}(s)=I$ for $0\leq s\leq t.$ Node $u\in {\cal C}(v_r)$ has two possible states $S$ or $R.$

{\bf Step 1.a} $u$ is susceptible if it was not infected within $t$ time slots. In each time slot, $v_r$ tries to infect $u$ with probability $q.$ The probability that $u$ is susceptible at time slot $t$ is
\[
(1-q)^t,
\] which implies that
\begin{align}
\Pr({\bf X}([0,t],T_u^{-v_r})| X_{v_r}(t)=I)=(1-q)^t \label{eq: k=0-S}
\end{align} if $X_u(t)=S.$

{\bf Step 1.b} If $u$ is in the recovered state, we denote by $t^I_u$ and $t^R_u$ its infection and recovery times, respectively. Then, we have if $X_u(t)=R,$
\begin{align*}
&\Pr({\bf X}([0,t],T_u^{-v_r})| X_{v_r}(t)=I)\\
&=(1-q)^{t^I_u-1}q(1-p)^{t^R_u-t^I_u-1}p\\
&\prod_{w \in {\cal C}(u)}\Pr\left({\bf
X}([0,t],T_w^{-u})|t^I_u,t^R_u\right),
\end{align*}
where $(1-q)^{t^I_u-1} q(1-p)^{t^R_u-t^I_u-1}p$ is the probability that node
$u$ was infected at time $t^I_u$, and recovered at time $t^R_u.$
Since $T_w^{-u}$ is also an infinite tree, there exists at least one
node $\xi \in T_w^{-u}$ such that the node is in the susceptible
state but its parent node (say node $\gamma$) is in the recovered
state. We denote by $T_w^{-u}\backslash T_{\xi }^{-\gamma}$ the set of nodes that are on subtree $T_w^{-u}$ but not on subtree $T_{\xi}^{-\gamma}.$ Then, \begin{align}
&\Pr\left({\bf X}([0,t],T_w^{-u})|t^I_u,t^R_u\right)\\
&=\Pr\left({\bf X}([0,t],T_w^{-u}\backslash T_{\xi
}^{-\gamma})|t^I_u,t^R_u\right)\\
&\times \Pr\left({\bf X}([0,t], T_{\xi
}^{-\gamma})|t^I_{\gamma},t^R_{\gamma}\right)\\
&=\Pr\left({\bf X}([0,t],T_w^{-u}\backslash T_{\xi
}^{-\gamma})|t^I_u,t^R_u\right)(1-q)^{t^R_{\gamma}-t^I_{\gamma}}\label{eqn:opt}\\
&\leq (1-q)\label{eqn:optrequirement},
\end{align}
where equation (\ref{eqn:opt}) holds because $\xi$ remained to be
susceptible during the time slots at which $\gamma$ was in the infected state and
(\ref{eqn:optrequirement}) holds because
$t^R_{\gamma}-t^I_{\gamma}\geq1.$ The maximum value of
$\Pr\left({\bf X}([0,t],T_w^{-u})|t^I_u,t^R_u\right)$ can be
achieved in the sample path in which $u$ was infected and then
recovered in the next time slot so that $w$ was
vulnerable to infection only in one time slot. Furthermore,
\[
(1-q)^{t^I_u-1}q(1-p)^{t^R_u-t^I_u-1}p
\]
is maximized when $t^{I}_u=1,t^{R}_u=2$ i.e., $u$ was infected
at the first time slot and recovered in the second time slot. Therefore, if $X_u(t)=R,$
\begin{align}
&\Pr({\bf X}([0,t],T_u^{-v_r})| X_{v_r}(t)=I) \leq qp(1-q)^{|{\cal C}(u)|}. \label{eq: k=0-R}
\end{align}

{\bf Step 1.c} Define ${\bf X}^*[0,t]$ to be the optimal solution to
\[
\max_{{\bf X}[0,t]\in {\cal X}(t)} \Pr({\bf X}[0,t]).
\]
For $t=1,$ since all $u\in {\cal C}(v_r)$ are in the susceptible state,
\begin{align}
\Pr\left({\bf X}^*([0,t]\right)=(1-p)(1-q)^{|{\cal C}(v_r)|}.
\end{align} For $t\geq 2,$ according to (\ref{eq: k=0-S}) and (\ref{eq: k=0-R}),
\begin{align}
&\Pr\left({\bf X}^*([0,t]\right)\\
&=(1-p)^t \prod_{u \in {\cal
C}(v_r)}\max\left\{(1-q)^{t},qp(1-q)^{|{\cal
C}(u)|}\right\}\label{eqn:t2}.
\end{align}
Note that $t$ is fixed in this optimization problem and
(\ref{eqn:t2}) is a none-increasing function of $t$. Since $|{\cal
C}(u)|\geq 1$,
\begin{align*}
&\Pr\left({\bf X}^*([0,2]\right)\\
&\leq (1-p)^2 \max\left\{(1-q)^{2|{\cal
C}(v_r)|},\left(qp(1-q)\right)^{|{\cal C}(v_r)|}\right\}\\
&< (1-p)(1-q)^{|{\cal C}(v_r)|}\\
&=\Pr\left({\bf X}^*[0,1]\right).
\end{align*}
In a summary, $\Pr\left({\bf X}^*([0,t]\right)$ is a none-increasing
function of $t\in[1,\infty)$ when $k=0.$

\begin{figure}
\begin{centering}
  \includegraphics[width=200 pt]{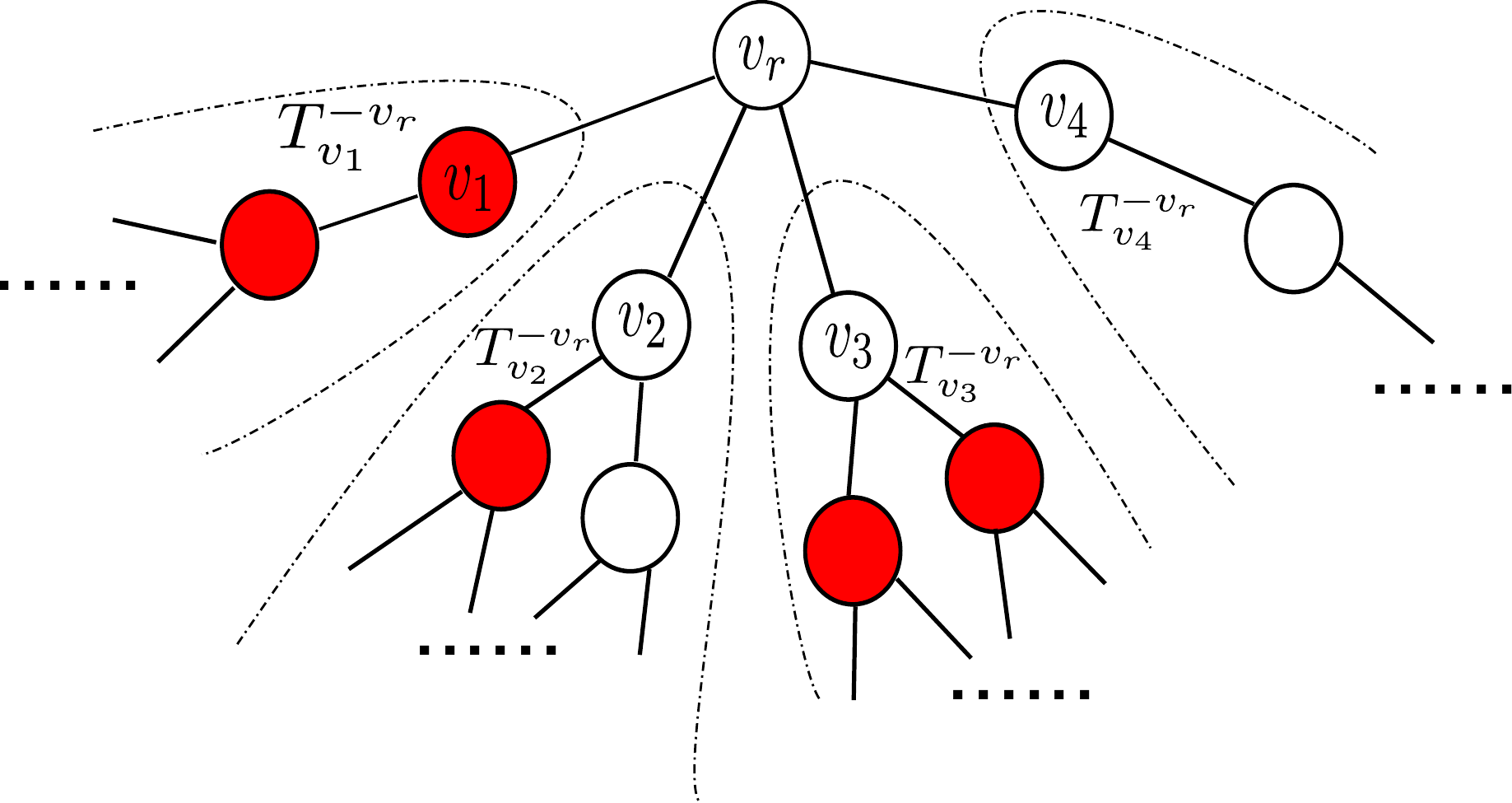}\\
  \caption{Example of Lemma \ref{lem:inequalitytime}}\label{figure:timeinequality}
  \end{centering}
\end{figure}
{\bf Step 2:} Assume (\ref{eqn:3}) holds for $k\leq n,$ and consider $k=n+1.$ Clearly $t\geq n+1\geq 1$ for each ${\bf X}[0,t]$ such that $${\bf F}({\bf X}[0,t])\in {\cal Y}^{n+1}.$$ Furthermore, the set of
subtrees ${\cal T}=\{T^{-v_r}_{u}|u \in {\cal C}(v_r)\}$ are
divided into two subsets: ${\cal T}^h=\{T^{-v_r}_{u}|u \in {\cal
C}(v_r), {\bf Y}(T^{-v_r}_{u})\cap{\cal I}=\emptyset\}$ and ${\cal
T}^i={\cal T}\backslash {\cal T}^h,$ where ${\bf Y}(T^{-v_r}_{u})$ is the vector of $\bf Y$ restricted to subtree $T^{-v_r}_u.$ In Figure 3, ${\cal
T}^h=\{T^{-v_r}_{v_4}\}$ and ${\cal
T}^i=\{T^{-v_r}_{v_1},\ T^{-v_r}_{v_2},\ T^{-v_r}_{v_3}\}.$ We note that given  $t_{v_r}^R,$ the infection
processes on the sub-trees are mutually independent.

{\bf Step 2.a} Recall that
${\cal T}^h$ is the set of subtrees having no infected nodes.
Following the argument for the $k=0$ case, we can obtain that if $T_u^{-v_r}\in {\cal T}^h,$ then
\begin{align*}
&\Pr({\bf X}^*([0,t],{T}^{-v_r}_u)|t_{v_r}^R)=\max\left\{(1-q)^{t_{v_r}^R},qp(1-q)^{|{\cal
C}(u)|}\right\}
\end{align*} when $t\geq t_{v_r}^R$ and
\begin{align*}
&\Pr({\bf X}^*([0,t],{T}^{-v_r}_u)|t_{v_r}^R)= \max\left\{(1-q)^{t},qp(1-q)^{|{\cal
C}(u)|}\right\}
\end{align*} when $t< t_{v_r}^R.$ So $\Pr({\bf X}([0,t],{T}^{-v_r}_u)|t_{v_r}^R)$ is non-increasing in $t$ given any $t_{v_r}^R.$

{\bf Step 2.b} For $T_u^{-v_r}\in {\cal T}^i$, given the sample path $\tilde{{\bf X}}([0,t+1], T_u^{-v_r}),$ we will
construct a sample path ${\bf X}([0,t], T_u^{-v_r})$ which occurs with a higher
probability. Denote the infection time of $u$ in sample path ${\bf X}([0,t], T_u^{-v_r})$  by $t^I_u$. We let $\tilde{t}^I_{u}$ denote the
infection time in sample path $\tilde{{\bf X}}([0,t+1], T_u^{-v_r}).$

If $\tilde{t}^I_{u}>1$, we choose $t^I_{u}=\tilde{t}^I_{u}-1$, i.e.,
$u$ is infected one time slot later in $\tilde{{\bf X}}[0,
t+1]$ than that in ${\bf X}[0,t].$  Assume the infection processes after $u$ was infected
are the same in the two sample paths ${\bf X}([0,t], T_u^{-v_r})$ and $\tilde{{\bf X}}([0,t+1], T_u^{-v_r}).$ Therefore, we have
$$\Pr(\tilde{{\bf X}}([0,t+1],T^{-v_r}_{u}))=(1-q)^{\tilde{t}^I_{u}-1}q\Pr(\tilde{{\bf
X}}([0,t+1],T^{-v_r}_{u})|\tilde{t}_{u}^I),$$ and $$\Pr({\bf
X}([0,t],T^{-v_r}_{u}))=(1-q)^{t^I_{u}-1}q\Pr({\bf
X}([0,t],T^{-v_r}_{u})|t^I_{u}).$$ where $\Pr({\bf
X}([0,t],T^{-v_r}_{u})|t^I_{u})$ is the probability of ${\bf
X}([0,t],T^{-v_r}_{u})$  after $u$ was infected. Since the sample paths
${\bf X}([0,t], T_u^{-v_r})$ and $\tilde{{\bf X}}([0,t+1], T_u^{-v_r})$ are the same after
$u$ was infected, we obtain
\[
\Pr({\bf X}([0,t],T^{-v_r}_{u})|t^I_{u})=\Pr(\tilde{{\bf
X}}([0,t+1],T^{-v_r}_{u})|\tilde{t}_{u}^I).
\]
Therefore, with $t^I_{u}=\tilde{t}^I_{u}-1$, we get
\begin{align*}
&\Pr(\tilde{{\bf X}}([0,t+1],T^{-v_r}_{u}))<\Pr({\bf
X}([0,t],T^{-v_r}_{u}))
\end{align*}
 If $\tilde{t}^I_{u}=1,$ we set
$t^I_{u}=\tilde{t}^I_{u}=1$.\footnote{Note that we cannot apply the same argument to $\tilde{t}^I_u>1$ because $t^I_u=\tilde{t}^I_u$ may not be feasible in a valid ${\bf X}[0, t].$} Based on the induction assumption,
for $k\leq n$ since ${\bf Y}(T^{-v_r}_{u})\in {\cal Y}^m, m\leq n,$ we have
\begin{align*}
&\max_{{\bf X}([0,t],T^{-v_r}_{u})\in {\cal X}(t,T^{-v_r}_{u})}
\Pr({\bf X}([0,t],T^{-v_r}_{u}))\\
&> \max_{\tilde{\bf X}([0,t+1],T^{-v_r}_{u}) \in {\cal
X}(t+1,T^{-v_r}_{u})} \Pr(\tilde{\bf X}([0,t+1],T^{-v_r}_{u})),
\end{align*} where ${\cal
X}(t,T^{-v_r}_{u})=\{{\bf X}([0,t],T^{-v_r}_{u}): {\bf F}({\bf X}([0,t],T^{-v_r}_{u}))={\bf Y}(T^{-v_r}_u) \}.$
Therefore, given any $\tilde{\bf X}([0,t+1],T^{-v_r}_{u}),$ we can always find a corresponding sample path ${\bf X}([0,t],T^{-v_r}_{u}),$ which occurs with a higher probability.

{\bf Step 2.c} Now we consider the sample path ${\bf X}^*[0,t+1]$ and denote by $\tilde{t}_{v_r}^R$ the recovery time of node $v_r$ in ${\bf X}^*[0,t+1]$. We now construct a sample path $\bar{\bf X}[0,t]$ as follows:
\begin{itemize}
\item If $\tilde{t}_{v_r}^R>t+1,$ i.e., $v_r$ is an infected node, then $\bar{t}_{v_r}^R>t,$ where $\bar{t}_{v_r}^R$ is the recovery time of $v_r$ in $\bar{\bf X}[0,t].$

\item If $\tilde{t}_{v_r}^R\leq t,$ we choose $\bar{t}_{v_r}^R=\tilde{t}_{v_r}^R.$

\item If $\tilde{t}_{v_r}^R= t+1,$ we choose $\bar{t}_{v_r}^R=t.$

\end{itemize}We further complete $\bar{\bf X}[0,t]$ by having optimal ones on ${\cal T}^h$ and constructing the ones in ${\cal T}^i$ following step 2.b. According to steps 2.a and 2.b, it is easy to verify that $\bar{\bf X}[0,t]$ occurs with a higher probability than ${\bf X}^*[0,t+1].$ Therefore,
we conclude that inequality
(\ref{eqn:3}) holds for $k=n+1,$ hence for any $k$ according to the principle of induction.

{\bf Step 3} Repeatedly applying inequality (\ref{eqn:3}), we obtain
that $t^*_{v_r}$ is the minimum amount of time required to produce the observed infection
topology. The minimum time required is equal to the maximum distance from
$v_r$ to an infected node. Therefore, the lemma holds.
\end{proof}

\subsection{The Sample Path Based Estimator}
After deriving $t_v^*$, we have a unique $t_v^*$ for each
$v\in {\cal V}$. The next lemma states that the optimal sample path starting from a node with a smaller infection eccentricity is more likely to occur.
\begin{lemma}\label{lem:mainlemma}
Consider a tree network with infinitely many levels. Assume the information source is the root, and the observed
infection topology is ${\bf Y}$ which contains at least one infected node. For $u,v \in {\cal V}$ such that $(u,v)\in
{\cal E}$, if $t^*_u>t^*_v,$ then
\begin{align*}
\Pr({\bf X}^*_u([0,t^*_u])) < \Pr({\bf X}^*_v([0,t^*_v])),
\end{align*}
where ${\bf X}^*_u[0, t_u^*]$ is the optimal sample path starting from node $u.$
\end{lemma}

\begin{proof}
Recall that $T_{v}$ denotes the tree rooted at $v$ and $T^{-v}_u$ denotes the tree rooted
at $u$ but without the branch from $v.$ Furthermore, ${\cal C}(v)$ is
the set of children of $v,$ and ${\bf X}([0,t],T_u^{-v})$ is the sample
path ${\bf X}[0,t]$ restricted to $T_u^{-v}.$

{\bf Step 1:} The first step is to show $t^*_u=t^*_v+1.$ First we
claim $T^{-u}_v \cap {\cal I}\neq \emptyset.$ Otherwise, all
infected node are on $T^{-v}_u.$ Since on a tree, $v$ can only reach
nodes in $T^{-v}_u$ through edge $(u,v),$ $t^*_v=t^*_u+1,$ which
contradicts $t^*_u>t^*_v.$

If $T^{-v}_u \cap {\cal I}\neq \emptyset,$ $\forall a \in T^{-v}_u
\cap {\cal I}$, we have $$d(u,a)=d(v,a)-1\leq t^*_v-1,$$ and
$\forall b \in T^{-u}_v \cap {\cal I}$, $$d(u,b)=d(v,b)+1\leq
t^*_v+1.$$ Hence, $$t^*_u\leq t_v^*+1,$$ which implies that
$$t^*_v<t^*_u\leq t^*_v+1,$$ i.e., $t^*_u=t^*_v+1.$

If $T_u^{-v}\cap{\cal I}=\emptyset,$ all infected nodes are in $T_v^{-u},$ so it is obvious $t^*_u=t^*_v+1.$

{\bf Step 2:} In this step, we will prove that $t^I_v=1$ on the
sample path ${\bf X}_u^*[0, t^*_u].$ If $t^I_v>1$ on ${\bf X}_u^*([0,t_u^*]),$ then
$$t^*_u-t^I_v=t^*_v+1-t^I_v<t^*_v.$$ Note that according to the
 definition of $t_u^*$ and $t_v^I,$ within $t_u^*-t_v^I$ time slots,
 node $v$ can infect all infected nodes on $T_v^{-u}.$
Since $t^*_u=t^*_v+1,$ the infected node farthest from node $u$ must be on $T_v^{-u},$ which implies that there exists a node $a \in
T^{-u}_v$ such that $d(u,a)=t^*_u=t^*_v+1$ and $d(v,a)=t^*_v.$  So node
$v$ cannot reach $a$ within $t^*_u-t^I_v$ time slots, which
contradicts the fact that the infection can spread from node $v$ to
$a$ within $t_u^*-t_v^I$ time slots along the sample path $X_u^*[0,
t_u^*].$ Therefore, $t^I_v=1.$

{\bf Step 3:} Now given sample path ${\bf X}^*_u[0,t^*_u]$, we
construct ${\bf X}_v[0,t^*_v]$ which occurs with a higher
probability. We divide the sample path ${\bf X}^*_u[0,t^*_u]$ into
two parts along subtrees $T_u^{-v}$ and $T_v^{-u}.$ Since $t_v^I=1,$
we have
\begin{align*}
&\Pr({\bf X}^*_u[0,t^*_u])\\
&=q\Pr\left({\bf
X}^*_u\left([0,t^*_u],T^{-u}_v\right)\Big|t^I_v=1\right)\Pr\left({\bf
X}^*_u\left([0,t^*_u],T^{-v}_u\right)\right),
\end{align*}
where $q$ is the probability that $v$ is infected at the first time
slot. Suppose in ${\bf X}_v[0,t^*_v]$, node $u$ was infected at the
first time slot, then
\begin{align*}
&\Pr({\bf X}_v[0,t^*_v])=\\
&q\Pr\left({\bf
X}_v\left([0,t^*_v],T^{-u}_v\right)\right)\Pr\left({\bf
X}_v\left([0,t^*_v],T^{-v}_u\right)\Big|t^I_u=1\right).
\end{align*}

For the subtree $T_v^{-u},$ given ${\bf X}^*_u\left([0,t^*_u],T^{-u}_v\right),$ in which
$t_v^I=1$, we construct the partial sample path ${\bf
X}_v\left([0,t^*_v],T^{-u}_v\right)$ to be identical to ${\bf X}^*_u\left([0, t_u^*], T_v^{-u}\right)$ except that all events occur one time slot earlier, i.e.,
$${\bf X}_v\left([0,t^*_v],T^{-u}_v\right)={\bf X}^*_u\left([1, t_u^*], T_v^{-u}\right).$$
This
is feasible because $t_v^*=t_u^*-1.$ Then
$$\Pr\left({\bf X}^*_u\left([0,t^*_u],T^{-u}_v\right)\Big|t^I_v=1\right)=\Pr\left({\bf
X}_v\left([0,t^*_v],T^{-u}_v\right)\right).$$

For the subtree $T_u^{-v}$, we construct ${\bf X}_v([0,t^*_v],T_u^{-v})$ such that
\begin{align*}
&{\bf X}_v([0,t^*_v],T_u^{-v})\in \\
&{\arg\max}_{\tilde{\bf X}([0,t^*_v],T_u^{-v})\in{\cal
X}(t^*_v,T_u^{-v})}\Pr\left(\tilde{\bf
X}\left([0,t^*_v],T^{-v}_u\right)\Big|t^I_u=1\right).
\end{align*} Based on Lemma \ref{lem:inequalitytime}, we have
\begin{align*}
&\max_{\tilde{\bf X}([0,t^*_v],T_u^{-v})\in{\cal X}(t^*_v,T_u^{-v})}\Pr\left(\tilde{\bf X}\left([0,t^*_v],T^{-v}_u\right)\Big|t^I_u=1\right)=\\
&\max_{\tilde{\bf X}([0,t^*_u-1],T_u^{-v})\in{\cal X}(t^*_u-1,T_u^{-v})}\Pr\left(\tilde{\bf X}\left([0,t^*_u-1],T^{-v}_u\right)\Big|t^I_u=1\right)\\
&> \max_{{\bf X}([0,t^*_u],T_u^{-v})\in{\cal
X}(t^*_u,T_u^{-v})}\Pr\left({\bf
X}\left([0,t^*_u],T^{-v}_u\right)\right).
\end{align*}
 Therefore, given the optimal sample
path rooted at $u$, we have constructed a sample path rooted at $v$ which
occurs with a higher probability. The lemma holds.
\end{proof}
Next, we give a useful property of the Jordan infection centers in
the following lemma.
\begin{lemma}\label{lem:numofcenters}
On a tree network with at least one infected node, there exist at most two
Jordan infection centers. When the network has two Jordan infection
centers, the two must be neighbors. \hfill{$\square$}
\end{lemma}
\begin{proof}
First, we claim if there are more than one Jordan infection centers,
they must be adjacent. Suppose $v,u\in{\cal V}$ are two
Jordan infection centers and $\tilde{e}(v)=\tilde{e}(u)=\lambda.$
Suppose $v$ and $u$ are not adjacent, i.e., $d(v,u)>1$. Then, there exists
$w \in{\cal V}$ such that $$d(w,u)=1,$$ and $$d(w,v)=d(v,u)-1,$$
i.e., $w$ is a neighbor of $u$ and is on the shortest path between
$u$ and $v.$ Note in a tree structure $w$ is unique.

If ${\cal I}\cap T^{-w}_u=\emptyset,$ then $\forall a \in {\cal I},$
$$d(w,a)=d(u,a)-1<d(u,a),$$ which contradicts the fact that $u$ is a
 Jordan infection center.

If ${\cal I}\cap T^{-w}_u\neq \emptyset.$ Since $\forall b \in {\cal
I}\cap T^{-w}_u,$
$$d(v,b)=d(v,w)+d(w,b),$$ i.e., $$d(w,b)=d(v,b)-d(v,w)\leq
\lambda-1.$$ On the other hand, since $\tilde{e}(u)=\lambda,$
$\forall h \in T^{-u}_w\cap {\cal I},$ $$d(w,h)=d(u,h)-1\leq
\lambda-1.$$ In a summary, $\forall h \in {\cal I},$ $$d(w,h)\leq
\lambda-1,$$ which contradicts the fact that the minimum infection
eccentricity is $\lambda.$

 Therefore all Jordan infection centers
must be adjacent to each other. However, suppose there exist $n$
infection eccentricity centers where $n>2$, they would form a
clique with $n$ nodes which contradicts the fact that the graph is a tree.
Therefore, there exist at most two adjacent Jordan infection
centers.
\end{proof}

Based on Lemma \ref{lem:mainlemma} and Lemma \ref{lem:numofcenters},
we finish this section with the following theorem.
\begin{thm}\label{thm:OSPD source}
Consider a tree network with infinitely many levels. Assume that the observed
infection topology ${\bf Y}$ contains at least one infected node.  Then
the source node associated with ${\bf X}^*[0, t^*]$ (the solution
to the optimization problem (\ref{eqn:optsamplepath})) is a Jordan infection center, i.e.,
\[
v^{\dag}=\arg\min_{v\in {\cal V}}\tilde{e}(v).
\]
\end{thm}
\begin{proof}
We assume the network has two Jordan infection centers: $w$ and $u,$ and assume
$\tilde{e}(w)=\tilde{e}(u)=\lambda$. The same argument works for the case where the network has only one
Jordan infection center.

Based on Lemma
\ref{lem:numofcenters}, $w$ and $u$ must be adjacent. We will show
for any $a\in{\cal V}\backslash\{w,u\},$ there exists a path from
$a$ to $u$ (or $w$) along which the infection
eccentricity strictly decreases.

{\bf Step 1:} First, it is easy to see from  Figure \ref{figure:OSPD} that $d(\gamma, w)\leq \lambda-1$  $\forall \gamma\in T_w^{-u}\cap {\cal
I}.$ We next show that there exists a node $\xi$ such that the
equality holds.

Suppose that $d(\gamma,w)\leq \lambda-2$ for any $\gamma \in T^{-u}_w\cap {\cal
I},$ which implies
$$d(\gamma,u)\leq \lambda-1\quad \forall \gamma \in T^{-u}_w\cap {\cal
I}.$$
Since $w$ and $u$ are both Jordan infection centers, we have $\forall \gamma \in T^{-w}_u\cap {\cal I},$
\begin{align*}
&d(\gamma,w)\leq \lambda\\
&d(\gamma,u)\leq \lambda-1.
\end{align*}
In a summary, $\forall \gamma \in {\cal I},$
$$d(\gamma,u)\leq \lambda-1.$$ This contradicts the fact that
$\tilde{e}(w)=\tilde{e}(u)=\lambda.$ Therefore, there exists $\xi
\in T^{-u}_w\cap {\cal I}$ such that
$$d(\xi,w)=\lambda-1.$$

\begin{figure}
\begin{centering}
  % Requires \usepackage{graphicx}
  \includegraphics[width=200 pt]{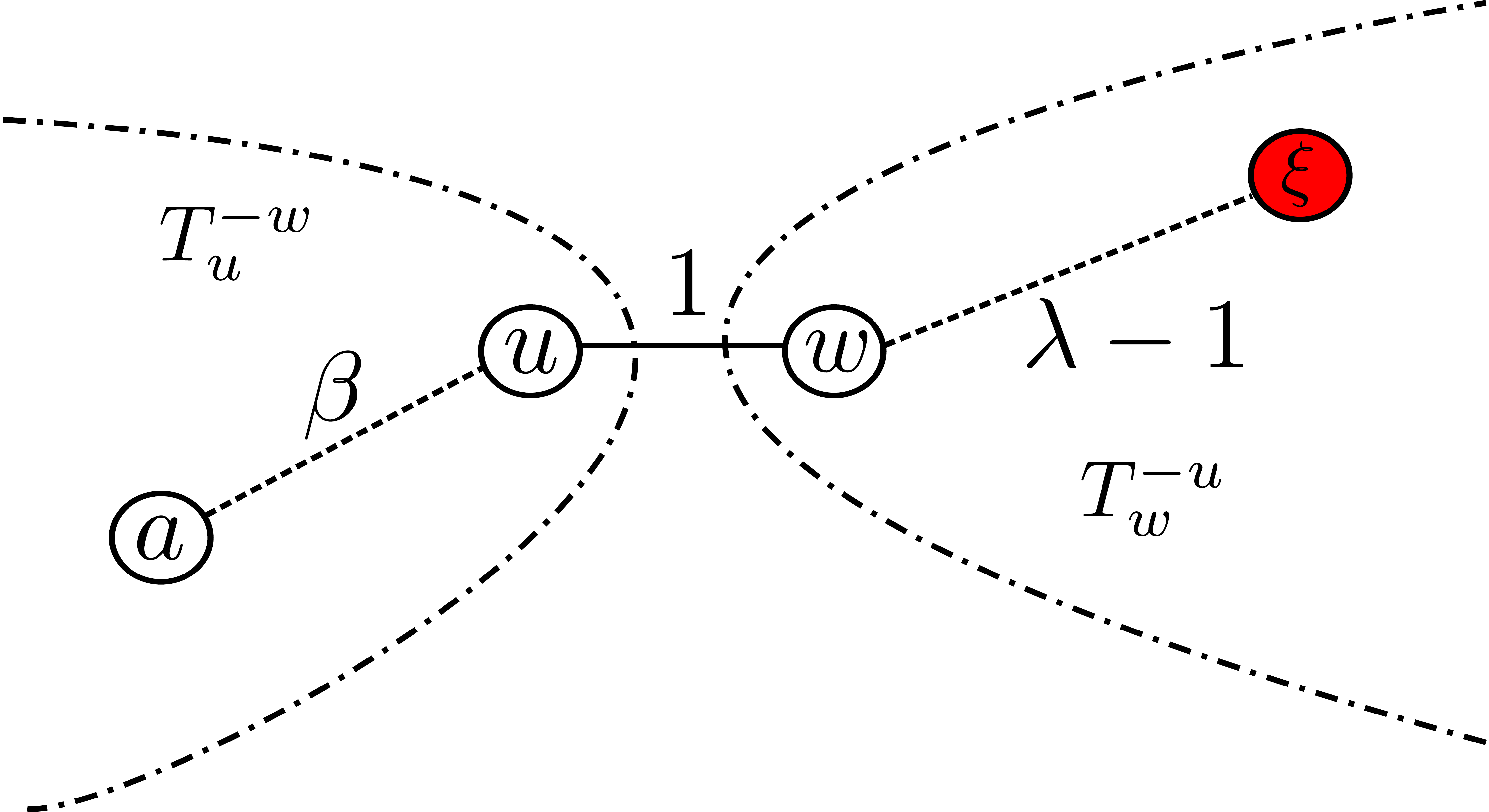}\\
  \caption{A Pictorial Description of the Positions of Nodes $a,$ $u,$ $w$ and $\xi.$}\label{figure:OSPD}
  \end{centering}
\end{figure}

{\bf Step 2:} Similarly, $\forall\gamma \in
T^{-w}_u\cap {\cal I},$ $$d(\gamma,u)\leq \lambda-1,$$ and
there exists a node such that the equality holds.

{\bf Step 3:} Next we consider $a\in{\cal V}\backslash\{w,u\},$ and
assume $a \in T^{-w}_u$ and $d(a,u)=\beta.$ Then for any $\gamma
\in T^{-u}_w\cap {\cal I},$ we have
\begin{align*}
d(a,\gamma)&=d(a,u)+d(u,w)+d(w,\gamma)\\
&\leq \beta+1+ \lambda-1\\
&=\lambda+\beta,
\end{align*}
and there exists $\xi\in T^{-u}_w\cap {\cal I}$ such that the equality holds. On the other
hand, $\forall \gamma \in T^{-w}_u\cap {\cal I}.$
\begin{align*}
d(a,\gamma)&\leq d(a,u)+d(u,\gamma)\\
&\leq \beta+ \lambda-1.
\end{align*}
Therefore, we conclude that
$$\tilde{e}(a)=\lambda+\beta,$$
 so the infection eccentricity decreases along the path from $a$ to
$u.$

{\bf Step 4:} Repeatedly applying Lemma \ref{lem:mainlemma} along the path from node $a$ to $u,$ we can
conclude that the optimal sample path rooted at node $u$ is more
likely to occur than the optimal sample path rooted at node $a.$ Therefore, the root
node associated with the optimal sample path ${\bf X}^*[0, t^*]$
must be a Jordan infection center, and the theorem holds.
\end{proof}

\section{Reverse Infection Algorithm}

Since in tree networks with infinitely many levels, the estimator based on the sample path
approach is a Jordan infection center, we view the Jordan
infection centers as possible candidates of the information source. We next present a simple algorithm to find the information source in general networks. The algorithm is to first identify the
Jordan infection centers, and then break ties based on the sum of distances to infected nodes.

The key idea of the algorithm is to let every infected node broadcast a
message containing its identity (ID) to its neighbors. Each node, after
receiving messages from its neighbors, checks whether the ID in the
message has been received. If not, the node records the ID (say
$v$), the time at which the message is received (say $t_v$), and then broadcasts
the ID to its neighbors. When a node receives the IDs of all
infected nodes, it claims itself as the information source and the
algorithm terminates. If there are multiple nodes receiving all
IDs at the same time, the tie is broken by selecting the
node with the smallest $\sum t_v.$

The tie-breaking rule we proposed is to choose the node with the
maximum infection closeness \cite{KosLehPee_05}. The
closeness measures the efficiency of a node to spread
information to all other nodes. The
closeness of a node is the inverse of the sum of
distances from the node to any other nodes. In our model, we define
the {\em infection closeness} as the inverse of the sum of distances from a node to all infected nodes,
which reflects the efficiency to spread information to infected
nodes. We select a Jordan infection center with the largest infection closeness, breaking ties at random.

\begin{algorithm}
\caption{Reverse Infection Algorithm}
\begin{algorithmic}
\FOR {$i\in{\cal I}$} \STATE $i$ sends its ID $\omega_i$ to its
neighbors. \ENDFOR

\WHILE {$t\geq 1$ and STOP$==0$}

\FOR{$u\in {\cal V}$}

\IF{$u$ receives $\omega_i$ for the first time}
 \STATE Set $t_{ui}=t$ and then broadcast the message $\omega_i$ to its
neighbors.
 \STATE If there exists a node who received $|{\cal I}|$
distinct messages, then set $STOP==1.$  \ENDIF
 \ENDFOR
\ENDWHILE

\RETURN $u^{\dag}=\arg\min_{u \in {\cal S}}\sum_{i \in {\cal
I}}t_{ui}$, where $\cal S$ is the set of nodes who receive $|\cal I|$
distinct messages when the algorithm terminates. Ties are broken at random.
\end{algorithmic}
\end{algorithm}
It is easy to verify that the set $\cal S$ is the set of the Jordan infection
centers. The running time of the algorithm is equal to the minimum
infection eccentricity and the number of messages each node
receives/sends during each time slot is bounded by its degree.

\section{Performance Analysis}\label{sec:performance}
The reverse infection algorithm is based on the structure properties
of the optimal sample paths on trees. While the MLE is the node that maximizes the likelihood of the snapshot among all possible nodes, the sample path based estimator does
not have such a guarantee. To demonstrate the effectiveness of
the sample path based approach, we next show that on $(g+1)$-regular
trees where each node has $g+1$ neighbors, the information source
generated by the reverse infection algorithm is within a constant
distance from the actual source with a high probability, independent of the number of infected nodes and the time at which
the snapshot ${\bf Y}$ was taken.

\begin{thm}\label{th:perfromance guarantee}
Consider a $(g+1)$-regular tree with infinitely many levels where $g>2$ and $gq>1.$ Assume that the observed
infection topology ${\bf Y}$ contains at least one infected node. Given $\epsilon>0$, there exists $d_{\epsilon}$ such that the distance between the optimal
sample path estimator and the actual source is $d_\epsilon$ with
probability $1-\epsilon,$ where $d_{\epsilon}$ is independent of the
number of infected nodes and the time the snapshot ${\bf Y}$ was taken.
\end{thm}

\begin{proof} Consider the tree rooted at the information source $v^*.$ We say $v^*$ is at level $0.$
We denote by ${\cal Z}_l$ the set of infected and recovered nodes at
level $l.$ Furthermore, we define ${\cal Z}_l^\tau$ to be the set of infected and recovered
nodes at level $l$ whose parents are in set ${\cal Z}_{l-1}^{\tau}$ and who were infected within $\tau$ time slots
after their parents were infected. We assume ${\cal Z}_0^{\tau}=\{v^*\}.$  In addition, let $Z_l=|{\cal
Z}_l|$ and $Z_l^\tau=|{\cal Z}_l^\tau|.$

 Note
$$\lim_{\tau\rightarrow \infty}Z_l^\tau=Z_l,$$
and given $v$ and $u\in {\cal Z}_l^\tau$,
$$|t^I_v-t^I_u|\leq l(\tau-1),$$ i.e., the infection
times of nodes in ${\cal Z}_l^\tau$ differ by at most $l(\tau-1)$ (note that the
difference is not $\tau-1$ since the parents of $u$ and $v$ may be
infected at different times). Our proof is based on the Galton
Watson (GW) branching process \cite{HacJagVat_05}. A GW branching
process is a stochastic process $B(l)$ which evolves according to
the recurrence formula $B(0)=1$ and
$$B(l)=\sum_{i=1}^{B(l-1)} \zeta_i,$$
where $\{\zeta_i\}$ is a set of random variables, taking
values from nonnegative integers. The distribution of $\zeta_i$ is
called the offspring distribution of the branching process. In a
$(g+1)$-regular tree, the evolution of ${\cal Z}_l^{\tau}$ is a
branching process, where the offspring distribution is a function of
$\tau.$ We use $B^{\tau}$ to denote the corresponding branching
process, and $B^{\tau}(l)$ to denote the number of offsprings at level $l,$
i.e., $B^{\tau}(l)=Z^{\tau}_l$ (we use these two notations
interchangeably). Given a node is in the infected state for $t$ time slots,
the number of infected offsprings follows a binomial
distribution. Note the following two facts:
\begin{itemize}
  \item The number of time slots at which a node is in the infected state follows a geometric distribution
  with parameter $p$.
  \item A child remains to be susceptible with probability $(1-q)^\tau$ when the parent has been in the infected state for $\tau$ time slot.
\end{itemize}
 Therefore, the offspring distribution of the branching process $B^{\tau}$ at level $\geq 1$\footnote{The source node has $g+1$ children while other nodes have $g$ children} is
\begin{align*}
&\Pr(\gamma=i)\\
&=\sum_{t=1}^{\tau-1}(1-p)^{t-1}p\binom{g}{i}\left(1-(1-q)^t\right)^{i}(1-q)^{t(g-i)}\\
&+\left(1-\sum_{t=1}^{\tau-1}(1-p)^{t-1}p\right)\binom{g}{i}\left(1-(1-q)^{\tau}\right)^{i}(1-q)^{\tau(g-i)},
\end{align*} where $\gamma$ is the number of offsprings of a node.
The offspring distribution of branching process $B^{\infty}$ is
\begin{align*}
&\Pr(\gamma'=i)\\
&=\sum_{t=1}^{\infty}(1-p)^{t-1}p\binom{g}{i}\left(1-(1-q)^t\right)^{i}(1-q)^{t(g-i)}.
\end{align*}

Each infected node can be viewed the source of branching processes on the subtree rooted at the node. We define $K_l$ to be the number of survived $B^1$ branching processes whose roots are in set ${\cal Z}_l^\tau,$ where a branching process survives if it never dies out.

Now given $L\geq 2,$ we consider the following events:
\begin{itemize}
  \item Event 1: $Z_L=0$
  \item Event 2:  $K_{l}\geq 2$ for some $l\leq L.$ In other words, at least two $B^1$ branching processes starting from ${\cal Z}_l^\tau$ survive for some $l\leq L.$
\end{itemize}
We note that these two are disjoint events.

When $Z_l=0,$ no node at level $L$ is infected and the infection
process terminates at level $L-1.$ When there is at least one infected node in $\bf Y,$ since $\tilde{e}(v^*)\leq L-1,$ the minimum infection eccentricity is at most $L-1.$ Therefore, the distance between
$v^*$ and $v^{\dag}$ is no more than $2(L-1).$

\begin{figure}[!t]
\begin{centering}
  \includegraphics[width=0.9\columnwidth]{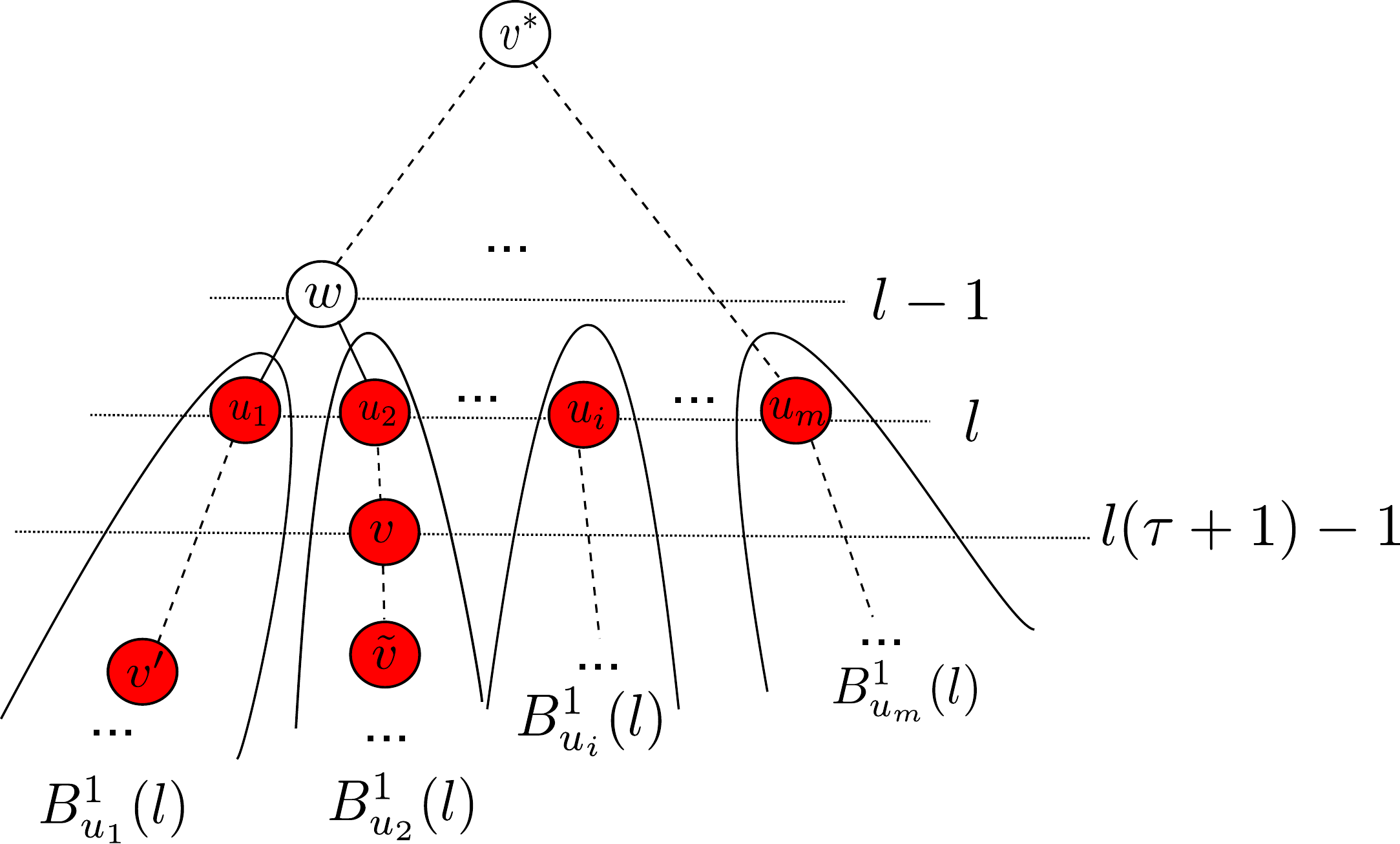}\\
  \caption{A pictorial description of the positions of $v^{\prime},$ $\tilde{v},$ $u_1,$ and $w.$}\label{figure:pgg}
\end{centering}
\end{figure}

Given $K_l\geq 2$ for some $l\leq L,$ we will argue that the distance between the sample path based estimator and the actual one is upper bounded by $(\tau+1) L-1.$ Consider Figure \ref{figure:pgg}, where the shaded nodes are
infected and recovered nodes. We will show that if two $B^1$
branching processes starting from $l\leq L$ survive, a node at level $\geq(\tau+1) L-1$ cannot be
a Jordan infection center. Recall that at time $t,$ the distance
between any infected node and the actual source is no more than $t,$
which implies the eccentricity of a Jordan infection center is $\leq t.$ Now
consider a node $\tilde{v}$ at level $\geq (\tau+1) l-1.$ Recall that at least two $B^1$
branching processes starting from level $l$ survive. Let $u_1\in {\cal Z}_l^\tau$ be the root of a survived $B^1$ branching process, and assume node $\tilde{v}$ is not on the subtree rooted at $u_1.$ Further, assume $v^{\prime}$ is
an infected node at the lowest level on sub-tree $T_{u_1}^{-w}.$
Since the branching process $B^1_{u_1}$ survives, the
infection process propagates one level lower at each time slot and node
$v^{\prime}$ is at level $l+t-t_{u_1}^I.$

From Figure \ref{figure:pgg}, it is easy to see that the distance between $v^{\prime}$ and $\tilde{v}$ is at least
$$t-t^I_{u_1}+2+(\tau+1) l-1-l=t-t^I_{u_1}+\tau l+1,$$ which occurs
when the first common predecessor of nodes $v^{\prime}$ and
$\tilde{v}$ is at $l-1$ level. Note that the common predecessor
cannot appear at level $\geq l$ since $\tilde{v}$ is not on
$T_{u_1}^{-w}.$ Since $u_1\in {\cal Z}^{\tau}_{l},$ the infection time of
node $u_1$ is no later than $\tau l,$ i.e., $t^I_{u_1}\leq \tau l.$ Therefore, the distance
between $v^\prime$ and $\tilde{v}$ is at least $t+1,$ which is
larger than $t.$ Hence, $v^{\prime}$ cannot be a Jordan infection center. Since $l\leq L,$ any node at or below level $(\tau+1) L-1$ cannot be a Jordan infection center. In a summary, if event 2 occurs, then we have $$d(v^*,v^{\dag})\leq (\tau+1) L-1.$$

We next show that given any $\epsilon,$ we can find sufficiently
large $\tau$ and $L,$ independent of $t$ and the number of infected nodes, such that the probability that either event 1 or event 2 occurs is at least
$1-\epsilon.$

Given $n_0>0$ and $\tau>0,$ we define $$l^\dag=\min \left\{l: Z^\tau_{l}>n_0\right\},$$ i.e., $l^\dag$ is the first level at which $B^{\tau}$ has more than $n_0$ nodes. We first have
\begin{align*}
&\Pr(Z_L=0)+\Pr\left(K_{l}\geq 2\hbox{ for some }l\leq L\right)\\
\geq &\Pr(Z_L=0)+\Pr\left(K_{l^\dag}\geq 2\hbox{ and }l^\dag\leq L\right)\\
=&\Pr(Z_L=0)+\Pr\left(l^\dag\leq L\right) \Pr\left(K_{l^\dag}\geq 2\Big|l^\dag\leq L\right)\\
=&\Pr(Z_L=0)+\Pr\left(\bigcup_{i=1}^{L}
\left\{Z^{\tau}_i>n_0\right\}\right) \Pr\left(K_{l^\dag}\geq 2\Big|l^\dag\leq L\right)\\
\geq&\left(1-\Pr\left(\bigcap_{i=1}^{L}\left\{0<Z^{\tau}_i\leq n_0\right\}\right)-\Pr\left(\bigcup_{i=1}^{L}\left\{Z^{\tau}_i=0\right\}\right)\right)\\
&\times \Pr\left(K_{l^\dag}\geq 2\Big|l^\dag\leq L\right)+\Pr(Z_L=0).
\end{align*}

Note that we have
\begin{align}
&\Pr(K_{l^\dag}\geq 2|l^\dag\leq L)\nonumber\\
=&\sum_{l=1}^L \Pr(K_{l^\dag}\geq 2|  l^\dagger=l) \Pr( l^\dagger=l|l^\dag\leq L). \label{eqn:KL}
\end{align}
According to Lemma \ref{lem:2survive}, given any $\epsilon_1>0,$ we can find a sufficiently large $n_0$ such that
\[
\Pr(K_{l^\dag}\geq 2|l^\dagger=l)\geq (1-\epsilon_1),
\]
which implies that for sufficiently large $n_0,$
\[
\Pr(K_{l^\dag}\geq 2|l^\dag\leq L)\geq 1-\epsilon_1.
\]
We can then conclude
\begin{align*}
&\Pr(Z_L=0)+\Pr\left(K_{l}\geq 2\hbox{ for some }l\leq L\right)\\
\geq& \left(1-\Pr\left(\bigcap_{i=1}^{L}\left\{0<Z^{\tau}_i\leq
n_0\right\}\right)\right)(1-\epsilon_1)\\
&-\Pr\left(\bigcup_{i=1}^{L}\left\{Z^{\tau}_i=0\right\}\right)+\Pr(Z_L=0)\\
=&\left(1-\Pr\left(\bigcap_{i=1}^{L}\left\{0<Z^{\tau}_i\leq
n_0\right\}\right)\right)(1-\epsilon_1)\\
&+\Pr(Z_L=0)-\Pr(Z_L^\tau=0),
\end{align*}
where $\Pr(\cup_{i=1}^L \left\{Z_i^\tau=0\right\})=\Pr(Z_L^\tau=0)$ because
$Z_l^{\tau}=0$ implies that $Z_L^{\tau}=0$ for $l\leq L.$

According to Lemma \ref{lem:morethann} and Lemma
\ref{lem:exinctcase}, given any $\epsilon_2>0$ and $\epsilon_3>0,$
there exist sufficiently large $\tau$ and $L$ such that
$$\left(1-\Pr\left(\bigcap_{i=1}^{L}\left\{0<Z^{\tau}_i\leq
n_0\right\}\right)\right)>1-\epsilon_2,$$ and
$$\Pr(Z_L=0)-\Pr(Z^\tau_L=0)\geq -\epsilon_3.$$
Hence, we have
\begin{eqnarray*}
&&\Pr(Z_L=0)+\Pr\left(K_{l}\geq 2\hbox{ for some }l\leq L\right)\\
&\geq&(1-\epsilon_1)(1-\epsilon_2)-\epsilon_3.
\end{eqnarray*} Now choosing $\epsilon_1=\epsilon_2=\epsilon_3=\epsilon_4/3$ for some $\epsilon_4>0,$ we have
\begin{eqnarray*}
&&\Pr(Z_L=0)+\Pr\left(K_{l}\geq 2\hbox{ for some }l\leq L\right)\\
&\geq&1-\epsilon_4.
\end{eqnarray*}

Now let $|{\bf Y}|$ denote the number of infected nodes in the observation $\bf Y.$ Define events $E_1=\{Z_L=0\}$ and $E_2=\{K_{l}\geq 2\hbox{ for some }l\leq L\}.$ We have
\begin{align*}
&\Pr(E_1||{\bf Y}|=1)+\Pr\left(E_2||{\bf Y}|=1\right)\\
=&\frac{1}{\Pr(|{\bf Y}|=1)}\left({\Pr(E_1\cap \{|{\bf Y}|=1\})}+{\Pr\left(E_2\cap \{|{\bf Y}|=1\}\right)}\right).
\end{align*}
Since $E_2$ implies that $|{\bf Y}|=1,$ we have
\begin{align*}
&\Pr(E_1||{\bf Y}|=1)+\Pr\left(E_2||{\bf Y}|=1\right)\\
=&\frac{1}{\Pr(|{\bf Y}|=1)}\left({\Pr(E_1\cap \{|{\bf Y}|=1\})}+{\Pr\left(E_2\right)}\right)\\
=&\frac{1}{\Pr(|{\bf Y}|=1)}\left(\Pr(E_1)-{\Pr(E_1\cap \{|{\bf Y}|=0\})}+{\Pr\left(E_2\right)}\right)\\
\geq &\frac{1}{\Pr(|{\bf Y}|=1)}\left(\Pr(E_1)-{\Pr(\{|{\bf Y}|=0\})}+{\Pr\left(E_2\right)}\right)\\
\geq &\frac{1}{\Pr(|{\bf Y}|=1)}\left({\Pr(\{|{\bf Y}|=1\})}-\epsilon_4\right)\\
= &1-\frac{\epsilon_4}{\Pr(|{\bf Y}|=1)}.
\end{align*}
Note that $\Pr(|{\bf Y}|=1)$ is a positive constant since the $B^1$ branching process starting from the information source survives with non-zero probability. The theorem holds by choosing $\epsilon_4=\epsilon \Pr(|{\bf Y}|=1).$
\end{proof}

\begin{lemma}\label{lem:2survive}
Consider $n_0$ i.i.d GW branching processes with a binomial offspring
distribution with parameters $g$ and $q$ such that $gq>1.$ Denote by $K$ the number of branching processes that
survive.  Given any
$\epsilon>0,$ if
\[
 n_0 \geq
\frac{8\log\frac{1}{\epsilon}}{1-\rho},
\]
 then
\[
\Pr(K\geq 2)\geq 1-\epsilon,
\] where $\rho$ is the extinction probability of the GW branching process. In the binomial case, $\rho$ is the smallest non-negative root of equation $\rho=(1-q+q\rho)^g.$ \hfill{$\square$}
\end{lemma}
\begin{proof}
The extinction probability of a GW branching process is denoted by
$\rho,$ which is the smallest none negative root of equation
$\rho=G(\rho)$ according to \cite{HacJagVat_05}, where $G(\rho)$ is the
moment generating function of offspring distribution. In the
binomial case we have $G(\rho)=(1-q+q\rho)^g.$ $\rho<1$ when $gq>1.$

We define a Bernoulli random variable $H_i,$ for the $i^{\rm th}$
branching process such that
\[
H_i=\left\{
      \begin{array}{ll}
        1, & \hbox{if the $i$th branching process survives;} \\
        0, & \hbox{otherwise.}
      \end{array}
    \right.
\]
So $K=\sum_{i=1}^{n_0}H_i,$ and $$E[K]=n_0(1-\rho).$$

According to the Chernoff bound \cite{MitUpf_05}, we have
\begin{align*}
\Pr(K\leq(1-\delta)(1-\rho)n_0)<
    e^{-\frac{(1-\rho)n_0\delta^2}{2}}.
\end{align*}
Choose $\delta=0.5.$ The Lemma holds if $$(1-\rho)n_0/2\geq 2,$$ and
$$(1-\rho)n_0/8\geq \log 1/\epsilon.$$
\end{proof}

\begin{lemma}\label{lem:morethann}
Given any $\epsilon>0,$ there exists a constant $L^{\prime}$ such
that for any $L\geq L^{\prime},$
\[
\Pr\left(\bigcap_{i=1}^{L}\{0<Z^{\tau}_i\leq n_0\}\right)\leq
\epsilon.
\]\hfill{$\square$}
\end{lemma}

\begin{proof}
Define $p_\tau$ to be the probability that a node infects at least
one of its children if it is in the infection state for $\tau$ time
slots. We have
  \begin{align*}
  p_\tau&=\sum_{t=1}^{\tau-1} (1-p)^{t-1}p(1-(1-q)^{gt})\\
&+(1-p)^{\tau-1}(1-(1-q)^{g\tau}),
  \end{align*}
and
\begin{align*}
&\Pr(0<Z^{\tau}_l\leq n_0|0<Z^{\tau}_{l-1}\leq n_0)\\
\leq &\Pr(Z_l^{\tau}>0|0<Z_{l-1}^{\tau}\leq n_0)\\
\leq & 1-(1-p_\tau)^{n_0},
\end{align*}
which implies that
\begin{align*}
&\Pr\left(\bigcap_{i=1}^{L}0<Z^{\tau}_i\leq n_0\right)\\
=&\Pr(0<Z^{\tau}_{L}\leq n_0|0<Z^{\tau}_{L-1}\leq
n_0)\\
&\times\Pr(0<Z^{\tau}_{L-1}\leq
n_0|0<Z^{\tau}_{L-2}\leq n_0)\cdots\\
&\times\Pr(0<Z^{\tau}_{2}\leq n_0|0<Z^{\tau}_{1}\leq n_0)\Pr(0<Z^{\tau}_{1}\leq n_0)\\
 \leq&(1-(1-p_\tau)^{n_0})^{L}.
\end{align*} The lemma holds by choosing $$L^{\prime}=\left\lceil\frac{\log \epsilon}{\log\left(1-(1-p_\tau)^{n_0}\right)}\right\rceil.$$
\end{proof}

\begin{lemma}\label{lem:exinctcase}
Given any $\epsilon,$ there exist $\tau^{\prime}$ and $L^{\prime}$
such that for any $\tau>\tau^{\prime}$ and $L>L^{\prime}$
\[
\Pr(Z_{L}=0)-\Pr(Z^{\tau}_{L}=0)\geq-\epsilon.
\] \hfill{$\square$}
\end{lemma}
\begin{proof}
Note the difference can be re-written as
\begin{align*}
&\Pr(Z_L=0)-\Pr(Z_L^{\tau}=0)\\
=&\left(\Pr(Z_L\!=\!0)-\Pr(Z_{\infty}\!=\!0)\right)\!+\!\left(\Pr(Z^{\tau}_{\infty}\!=\!0)\!-\!\Pr(Z^{\tau}_{L}\!=\!0)\right)\\
+&\left(\Pr(Z_{\infty}=0)-\Pr(Z^{\tau}_{\infty}=0)\right).
\end{align*}

{\bf Step 1} Since $\{Z^{\tau}_{L}=0\}\subseteq \{Z^{\tau}_{\infty}=0\}$,
$$\Pr(Z^{\tau}_{\infty}=0)-\Pr(Z^{\tau}_{L}=0)\geq 0.$$

{\bf Step 2} We
know
\[
\lim_{L \rightarrow \infty}\Pr(Z_{L}=0)=\Pr(Z_{\infty}=0).
\]
 Then for $\epsilon/2>0$, there exists $L^\prime$ such that for any $L\geq L^\prime,$
\[
|\Pr(Z_{L}=0)-\Pr(Z_{\infty}=0)|\leq \epsilon/2,
\]
which implies that
\[
\Pr(Z_{L}=0)-\Pr(Z_{\infty}=0)\geq -\epsilon/2.
\]

{\bf Step 3} In this step, we will show \[
\lim_{\tau \rightarrow \infty} \Pr(Z^{\tau}_{\infty}=0)=
\Pr(Z_{\infty}=0).
\]
Define the generating functions of the offspring
distributions of $B^\tau$ and $B^{\infty}$ to be $G_{\tau}(s)$ and
$G(s)$, respectively. We know that $G_{\tau}(s)-s$ and $G(s)-s$ are
convex functions when $s \in [0,1]$. Let $\rho=\Pr(Z_\infty=0),$
i.e., the extinction probability, we know that $\rho$ is the
smallest nonnegative root of $G(\rho)=\rho$ and $\rho<1$.

Similarly, define $\rho_\tau=\Pr(Z^\tau_\infty=0),$ and $\rho'=\lim_{\tau \rightarrow \infty} \rho_{\tau}$. Taking limit on
both sides of $G_{\tau}(\rho_{\tau})=\rho_{\tau}$, we have
\[
G(\rho')=\rho'.
\]
Note that $$\rho\leq\rho_{\tau}\leq \rho_1<1$$ for any $\tau,$
so $$\rho\leq\rho'\leq\rho_1<1.$$ Since $G(s)-s=0$ has at most
two solutions in $[0,1]$ and $s=1$ is one of them, we conclude
$\rho'=\rho$. Therefore, for given $\epsilon/2>0$, there exists
$\tau\geq\tau^\prime $ such that
\[
\Pr(Z_{\infty}=0)-\Pr(Z^{\tau}_{\infty}=0)\geq -\epsilon/2.
\]
Hence, the lemma holds.
\end{proof}
\section{Simulations}
In this section, we evaluate the performance of the reverse
infection algorithm on different networks, including different tree networks and some real world networks.
\subsection{Tree Networks}
In this section, we evaluate the performance of the reverse
infection algorithm on tree networks. We compare the reverse infection algorithm with the closeness centrality heuristic, which
selects the node with the maximum infection closeness as the
information source. Note that the node with the maximum closeness
is the maximum likelihood estimator of the information
source on regular trees under the SI model \cite{ShaZam_10,ShaZam_11,ShaZam_12}.

\subsubsection{Small-size tree networks}
We first studied the performance on small-size trees.  The infection probability $q$ was
chosen uniformly from $(0,1)$ and the recovery probability $p$ was
chosen uniformly from $(0,q).$ The infection process propagates $t$
time slots where $t$ was uniformly chosen from $[3,5].$ To keep the
size of infection topology small, we restricted the total number of
infected and recovered nodes to be no more than $100.$ For small-size trees, we first calculated the MLE using dynamic programming for fixed $t$ and then searching over $t\in [0, t_{\max}]$
for a large value of $t_{\max}$ to find the optimal estimator.
\begin{figure}[htb]
\begin{centering}
  % Requires \usepackage{graphicx}
  \includegraphics[width=0.90\columnwidth]{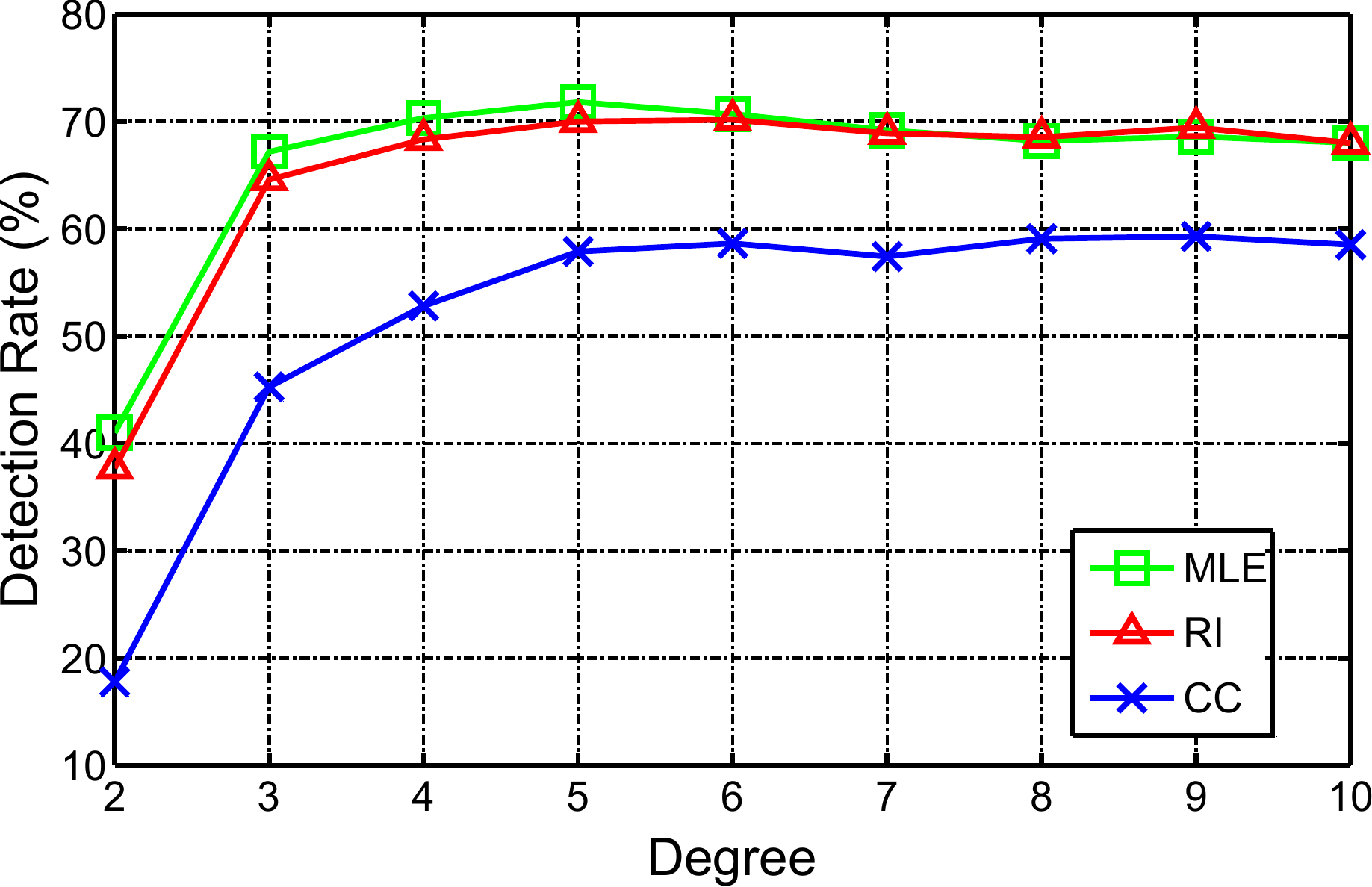}\\
  \caption{The Detection Rates of the Maximum Likelihood Estimator (MLE),
Reverse Infection (RI) and Closeness Centrality (CC) on Regular
Trees}\label{figure:smallregular}
  \end{centering}
\end{figure}

The detection rate is defined to be the fraction of experiments in
which the estimator coincides with the actual source. We varied $g$ from $2$ to
$10$ and the results are shown in Figure
\ref{figure:smallregular}. We can see that the detection rate of the
reverse infection algorithm is almost the same as that of the
MLE, and is higher than that of the closeness
centrality heuristic by approximately $20\%$ when the degree is
small and by $10\%$ when the degree is large.

\begin{figure}[htb]
\begin{centering}
  \includegraphics[width=0.90\columnwidth]{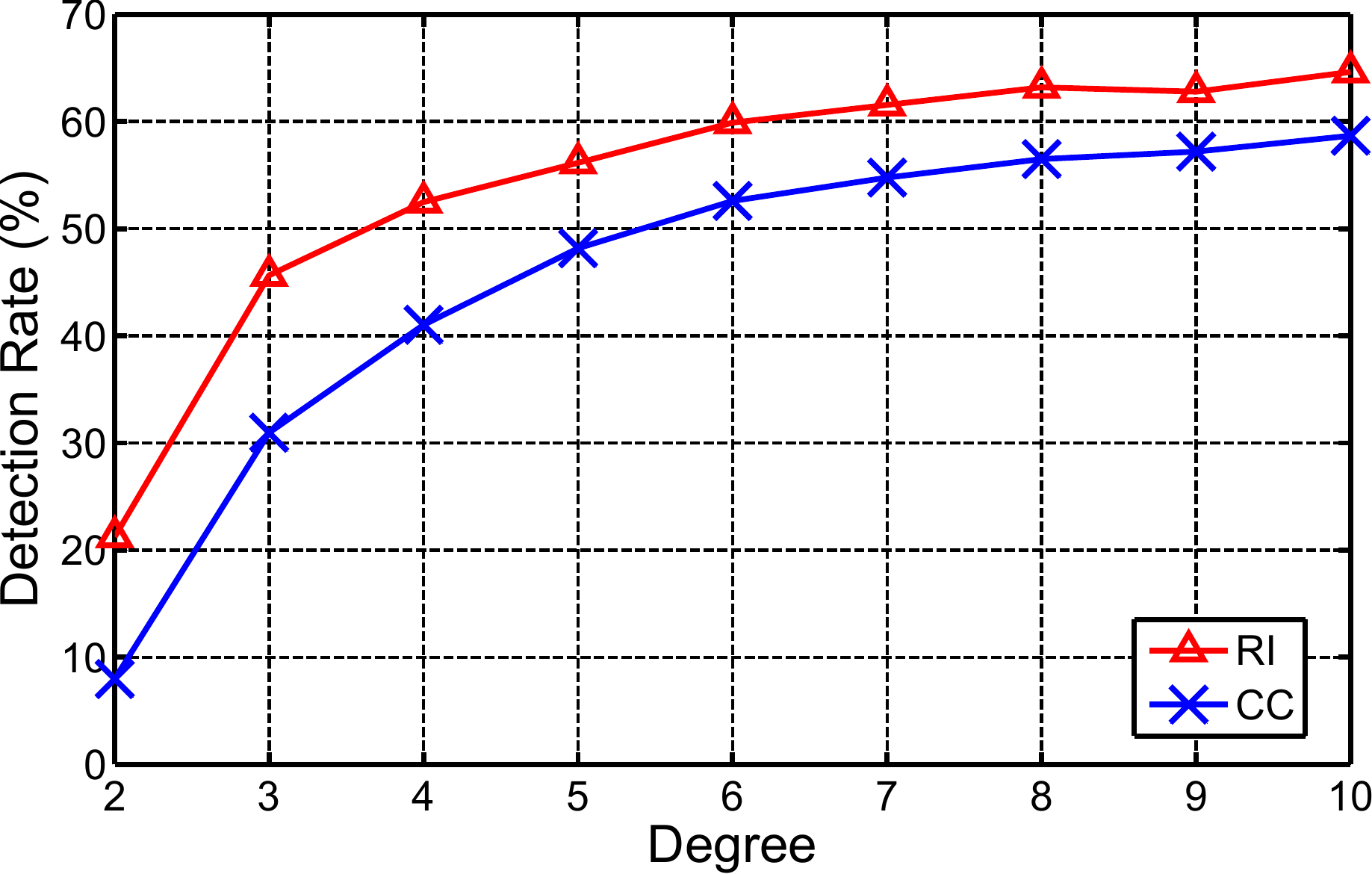}\\
  \caption{The Detection Rates of the
Reverse Infection (RI) and Closeness Centrality (CC) Algorithms on
Regular Trees}\label{figure:regular}
  \end{centering}
\end{figure}
\subsubsection{General $g$-regular tree networks}
We further conducted our simulations on large-size $g$-regular
trees. The infection probability $q$ was chosen uniformly from $(0,1)$ and
the recovery probability $p$ was chosen uniformly from $(0,q).$ The
infection process propagates $t$ time slots where $t$ was uniformly
chosen from $[3,20].$ We selected the networks in which the total
number of infected and recovered nodes is no more than $500.$

We varied $g$ from 2 to 10. Figure \ref{figure:regular} shows the
detection rate as a function of $g$. We can see the detection rates
of both the reverse infection and closeness centrality algorithms
increase as the degree increases and is higher than $60\%$ when
$g>6.$ However, he detection
rate of the reverse infection algorithm is higher than that of the
closeness centrality algorithm, and the average difference is
$8.86\%.$
\subsubsection{Binomial random trees}
In addition, we evaluated the performance on binomial random trees
where the number of children of each node follows a binomial
distribution with number of trials $g^\prime$ and success
probability $\beta.$ We fixed $g^\prime=10$ and varied $\beta$ from
0.1 to 0.9. The durations of the infection process and the observed
infected networks were selected according to the same rules for the
$g$-regular tree case.
\begin{figure}[!t]
\begin{centering}
  \includegraphics[width=0.9\columnwidth]{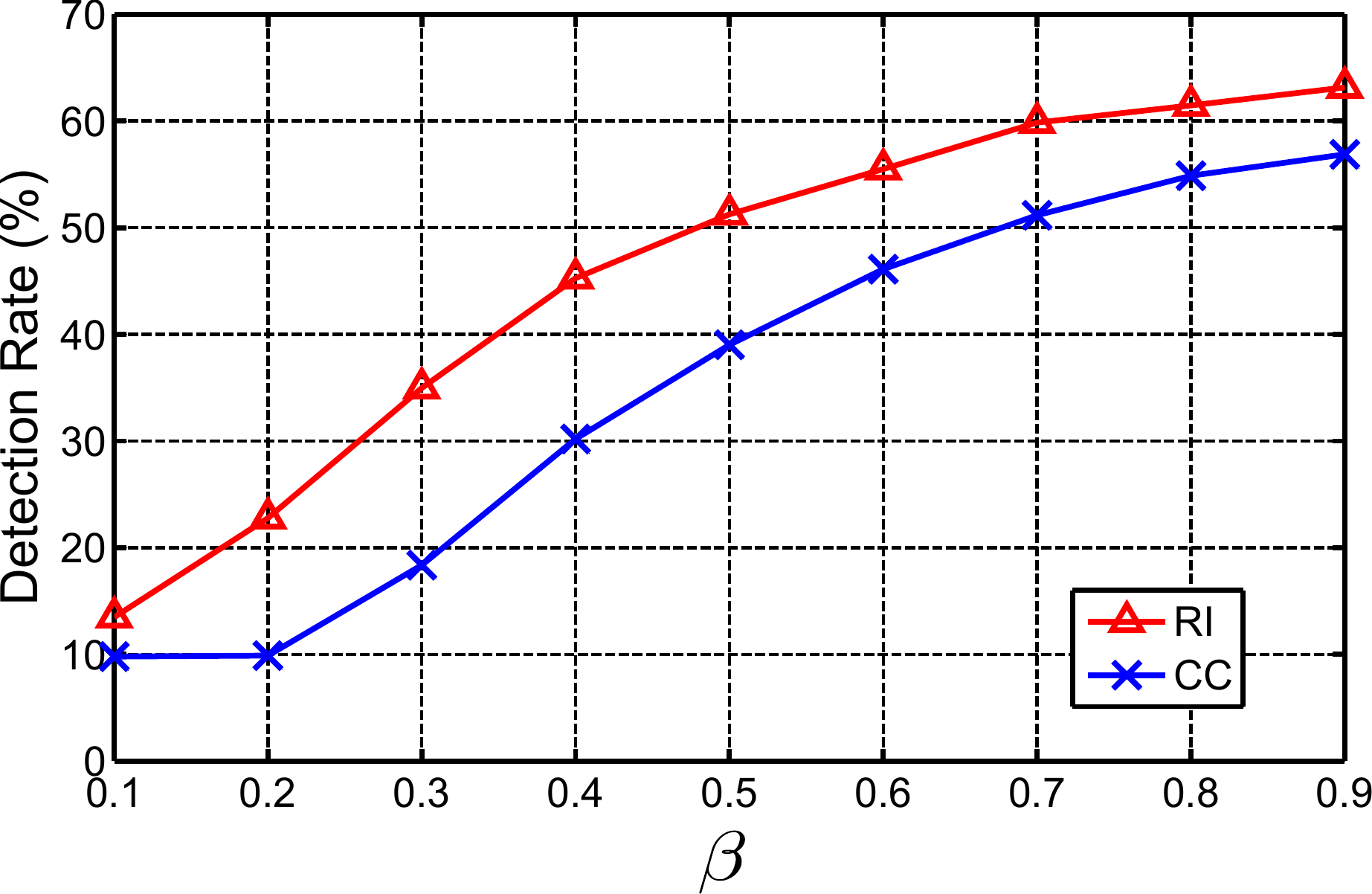}\\
  \caption{The Detection Rates of the Reverse Infection (RI) and Closeness Centrality (CC)
Algorithms on Binomial Random Trees}\label{figure:binomial}
  \end{centering}
\end{figure}
The results are shown in Figure \ref{figure:binomial}. Similar to
the regular tree case, as $\beta$ increases, the tree is more denser
which increases the number of survived branching processes and the
detection rate. The reverse infection algorithm outperforms the
closeness centrality algorithm by $10.16\%$ on average.
\begin{figure}[!t]
\begin{centering}
  \includegraphics[width=1\columnwidth]{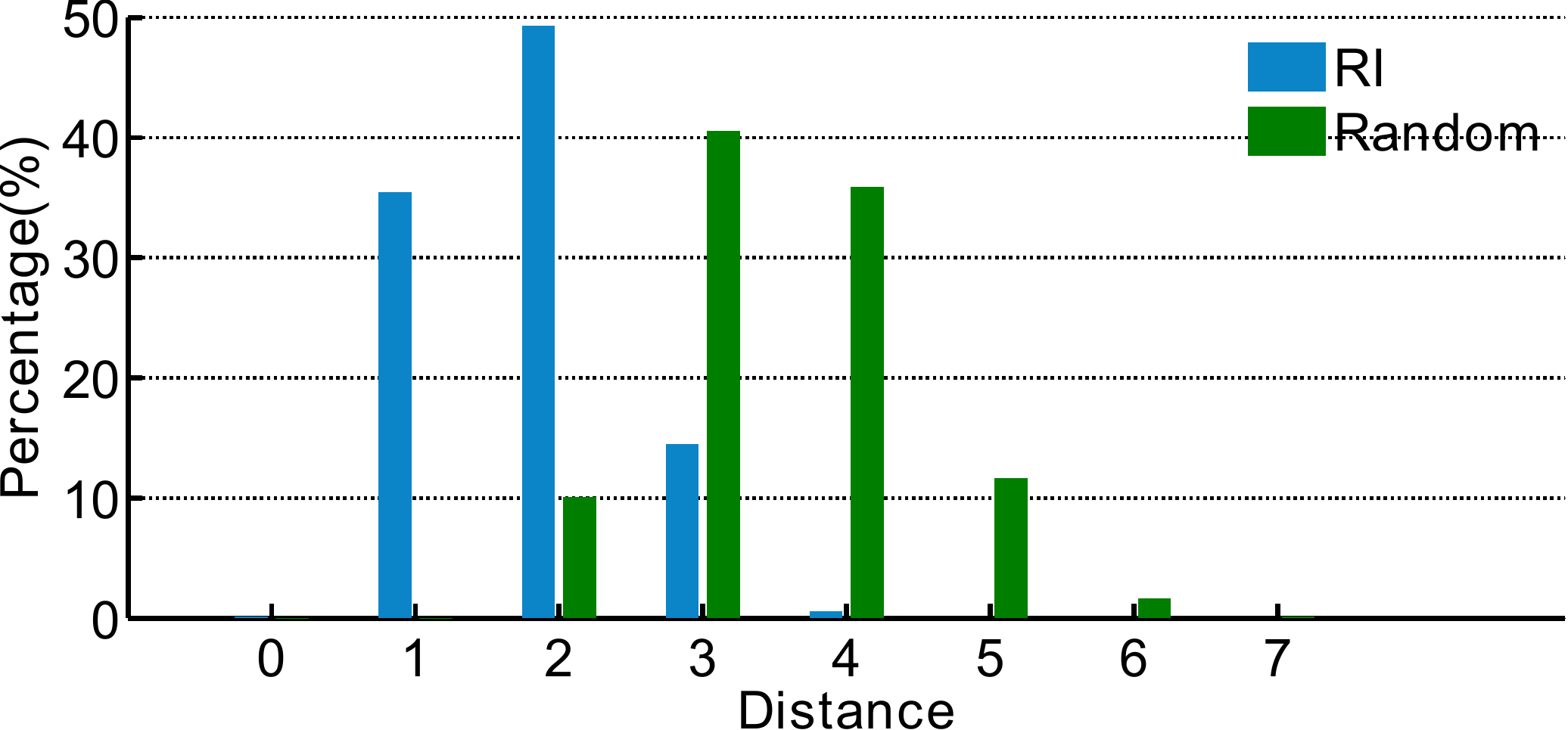}
  \caption{The Performance of the
Reverse Infection (RI) on the Internet Autonomous Systems Network}\label{figure:ias-total}
  \end{centering}
    \end{figure}

\subsection{Real World Networks}
We next conducted experiments on three real world networks
---  the Internet Autonomous Systems network
(IAS)\footnote{Available at
\url{http://snap.stanford.edu/data/index.html}}, the Wikipedia
who-votes-on-whom network (Wikipeida)\footnotemark[4], and the power
grid network (PG)\footnote{Available at
\url{http://www-personal.umich.edu/~mejn/netdata/}}. We compare the reverse infection algorithm with random guessing, which randomly selects a node and declares it as
the information source. In these networks, the infection
probability $q$ was chosen uniformly from $(0,0.05)$ and the
recovery probability $p$ was chosen uniformly from $(0,q).$ Here we
chose small infection probabilities since the network was of finite
size so the infection process should be controlled to make sure that
not all nodes were infected when the network was observed. The
duration $t$ was an integer uniformly chosen from $[3,200].$ We
selected the networks in which the total number of infected and
recovered nodes was in the range of $[50,500].$
\begin{figure}[!t]
\begin{centering}
  \includegraphics[width=1\columnwidth]{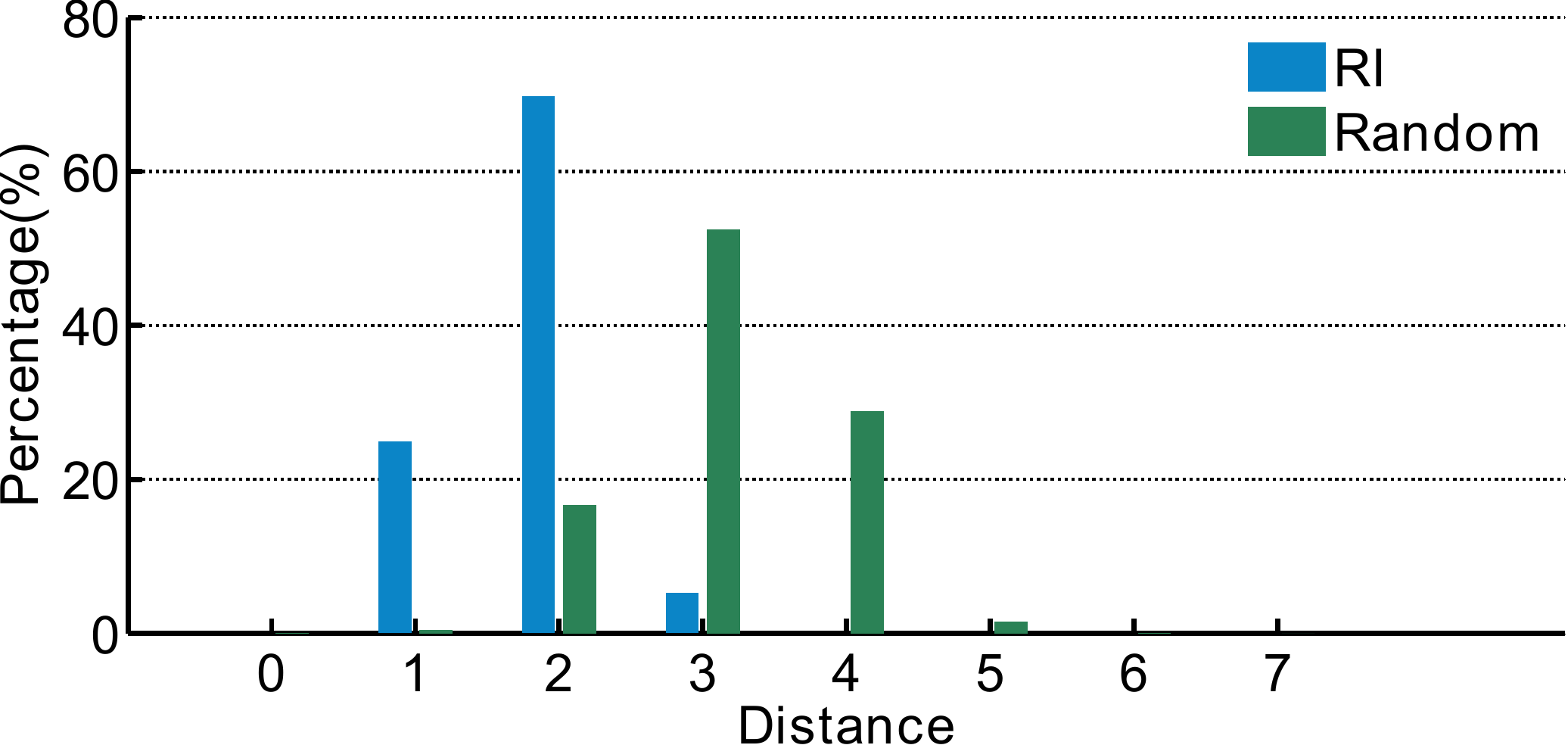}
  \caption{The Performance of the
Reverse Infection (RI) on the Wikipedia Who-Votes-on-Whom Network}\label{figure:wiki-total}
  \end{centering}
    \end{figure}

\subsubsection{The Internet autonomous systems network}
Figure \ref{figure:ias-total} shows the
results on the the Internet autonomous systems network. An Internet
autonomous system is a collection of connected routers who use a common routing
policy. The Internet
autonomous system network is obtained based on the recorded communication between the Internet autonomous systems
inferred from Oregon route-views on March, 31st, 2001. The network
consists of 10,670 nodes and 22,002 edges. According to Figure
\ref{figure:ias-total}, more than 80\% of the estimators identified
by the reverse infection algorithm are no more than two hops away
from the actual sources, comparing to 10\% under the random
guessing.
\begin{figure}[!t]
\begin{centering}
  \includegraphics[width=0.9\columnwidth]{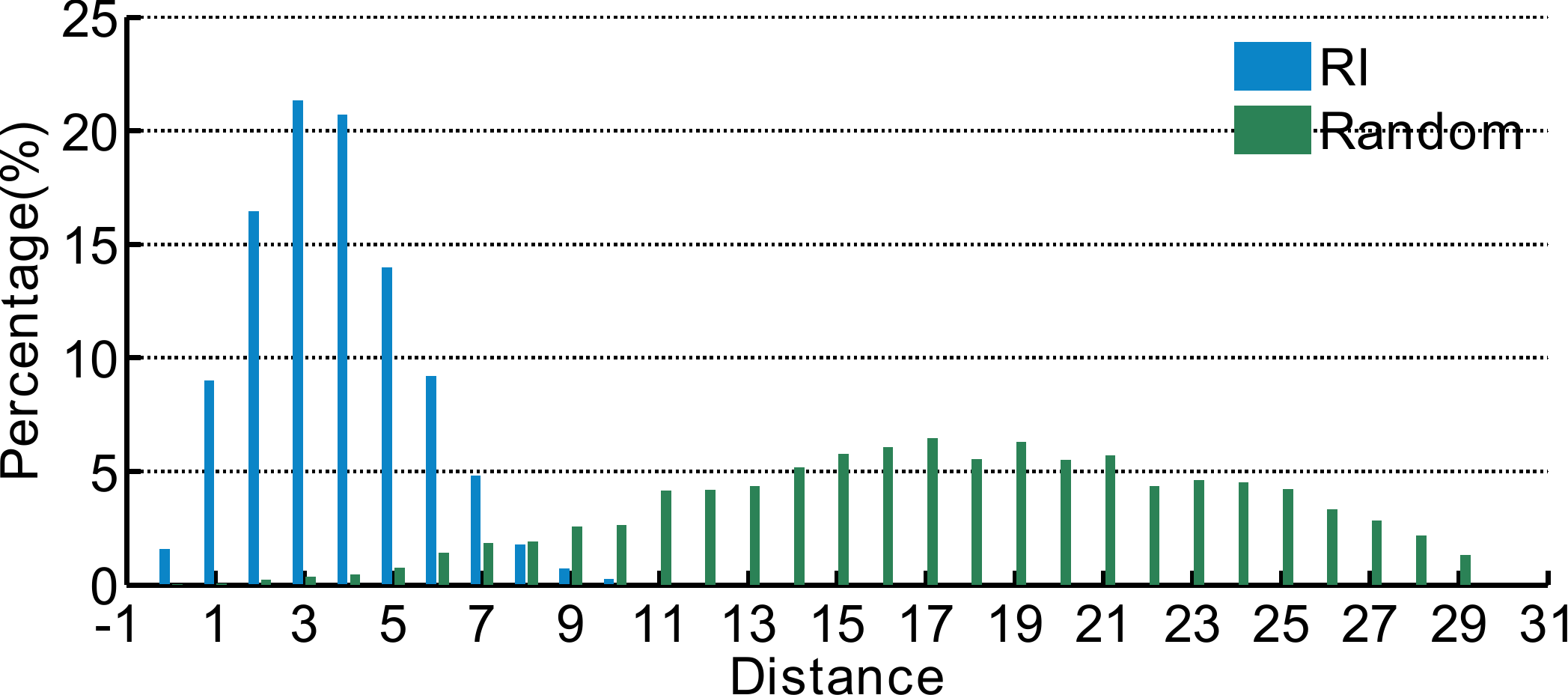}
  \caption{The Performance of the
Reverse Infection (RI) on the Power Grid Network}\label{figure:pg-total}
  \end{centering}
  \end{figure}

\subsubsection{The Wikipedia who-votes-on-whom network}
Figure \ref{figure:wiki-total} shows results on the
Wikipedia who-votes-on-whom network, in which two nodes are
connected if one user voted on the other in the administrator
promotion elections. The network has 100,736 links and 7,066 nodes.
We have similar observations as for the Internet autonomous systems network: the majority of the estimators produced by the reverse infection algorithm are no more than two hops away from the actual sources;
and only less than $20\%$ of the estimators of random guessing are within two hops from the actual sources.

\subsubsection{The power grid network}

Figure \ref{figure:pg-total} shows the
results on the power grid, which has 4,941 nodes and 6,594 edges. As
we can see, the reverse infection algorithm performs better than the
random guessing. The peak of the reverse infection algorithm appears
at the third hop versus the seventeenth hop under random
guessing.

\section{Conclusion}
In this paper, we developed a sample path based approach to find the information source under the SIR model. We
proved that the sample path based estimator is a node with the minimum infection eccentricity.
Based on that, a reverse infection algorithm has been proposed. We
analyzed the performance of the reverse infection algorithm on
regular trees, and showed that with a high probability the distance
between the estimator and actual source is a constant, independent of
the number of infected nodes and the time the network was observed.
We evaluated
the performance of the proposed reverse infection algorithm on
several different network topologies.

\begin{appendices}
\section{Notation Table}
\begin{tabularx}{0.5\textwidth}{c|X}
  \hline
  $q$ & the probability an infected node infects its neighbors\\
  \hline
  $p$ & the probability an infected node recovers  \\
  \hline
  $v^*$ & the actual information source \\
  \hline
  $v^{\dag}$ & the estimator of the information source \\
  \hline
  $d(v,u)$ & the length of shortest path between node $v$ and node $u$ \\
  \hline
  ${\cal C}(v)$ & the set of children of node $v$\\
  \hline
  $\tilde{e}(v)$ & the infection eccentricity of node $v$\\
  \hline
  $t^*_v$ & the time duration associated of the optimal sample path in which node $v$ is the information source \\
  \hline
  $t^I_v$ & the infection time of node $v$ \\
  \hline
  $t^R_v$ & the recovery time of node $v$ \\
  \hline
  ${\bf Y}$ & the snapshot of all nodes\\
  \hline
  $T_u^{-v}$ & the tree rooted at node $u$ but without the branch
from $v.$\\
\hline
$X_v(t)$ & the state of node $v$ at time $t$\\
\hline
${\bf X}(t)$ & the states of all nodes at time $t$\\
\hline
${\bf X}[0,t]$ & the sample path from $0$ to $t$\\
\hline ${\bf X}([0,t],T_u^{-v})$ & the sample path from time slot
$0$ to $t$
restricted to $T_u^{-v}$\\
\hline
${\cal X}(t)$ & the set of all valid sample path from time slot $0$ to $t$\\
\hline
${\cal X}(t,T_u^{-v})$ & the set of all valid sample path from time slot $0$ to $t$ restricted to $T_u^{-v}$\\
\hline
${\cal I}$ & the set of the infected nodes\\
\hline
\end{tabularx}
\end{appendices}
\balance
\bibliographystyle{IEEEtran}
\bibliography{inlab-refs}
\end{document}